\documentclass[11pt,oneside]{article}

\usepackage[round,colon,authoryear]{natbib}

\usepackage{graphicx}
\usepackage{amsmath}
\usepackage{amsfonts}
\usepackage{amssymb}
\usepackage{amsthm}
 \usepackage{mathrsfs}
\usepackage{enumerate}
\usepackage{pdfsync}
 \usepackage{stmaryrd}
\usepackage{comment}
\usepackage{color}
\usepackage{bm}
\usepackage{times}
\usepackage[inline]{enumitem}   

\usepackage{bm}

\usepackage{tikz}

\bmdefine\taub{\tau}
\bmdefine\mub{\mu}
\bmdefine\lab{\lambda}
\bmdefine\varsigmab{\varsigma}

 \numberwithin{equation}{section}
    
\newcommand{\R}{\mathbb{R}}
\newcommand{\N}{\mathbb{N}}
\newcommand{\1}{\mathbf{1}}

\bmdefine\thetab{\vartheta}

\newcommand{\filtration}[1]{\mathfrak{#1}}
\newcommand{\sigalgebra}[1]{\mathscr{#1}}
\renewcommand{\P}{\textsf{\upshape P}}
\newcommand{\E}{\textsf{\upshape E}}
\newcommand{\Qu}{\textsf{\upshape Q}}
\newcommand{\transpose}[1]{{#1}^\prime}

\newtheorem{thm}{Theorem}[section]
\newtheorem{lem}[thm]{Lemma}
\newtheorem{prop}[thm]{Proposition}

\theoremstyle{definition}
\newtheorem{defn}[thm]{Definition}

\theoremstyle{remark}
\newtheorem{example}[thm]{Example}

\theoremstyle{remark}
\newtheorem{rem}[thm]{Remark}

\title{
Trading Strategies   Generated by Lyapunov Functions
  \thanks{We are grateful to Robert Fernholz for initiating this line of research and  for   encouraging  us to think about the issues studied here. Many discussions with Kostas Kardaras helped us sharpen  our thoughts. We are also deeply indebted to Adrian Banner, Christa Cuchiero, Freddy Delbaen,   David Hobson, Tomoyuki Ichiba, Philip Protter, Mathieu Rosenbaum,   Walter Schachermayer, Konrad Swanepoel, Kangjia'Nan Xie, and Hao Xing    for helpful comments, and Alexander Vervuurt
 and Minghan Yan    for their   detailed  reading and suggestions  on successive versions of this paper. I.K.~acknowledges the support of  the National Science Foundation under  grant NSF-DMS-14-05210.   J.R.~acknowledges   generous   support  from  the Oxford-Man Institute of Quantitative Finance, University of Oxford.}
 }

\vspace{2mm}

\author{  
\textsc{Ioannis Karatzas} \thanks{
Department of Mathematics,  Columbia University, New York, NY 10027 (E-mail: {\it ik@math.columbia.edu}), and       \textsc{Intech} Investment Management,  One Palmer Square, Suite 441, Princeton, NJ 08542    (E-mail:    {\it ikaratzas@intechjanus.com}). 
}  
 \and
\textsc{Johannes Ruf}                \thanks{ 
Department of Mathematics, University College London, Gower Street, London WC1E 6BT, United Kingdom (E-mail:    {\it j.ruf@ucl.ac.uk}).           }
                                      }

\begin{document}

\maketitle

~~~~~~~~~~~~~~~~~ {\it Dedicated to Dr.~E.~Robert Fernholz on the occasion of his 75th Birthday}

 \bigskip

\begin{abstract}
\noindent
Functional portfolio generation,  initiated by E.R. Fernholz almost twenty years ago, is a methodology for constructing trading strategies with controlled behavior. It is based on very weak and descriptive    assumptions on the covariation structure of the underlying market model, and    needs no estimation of model parameters.    In this paper, the corresponding   generating functions $G$  are interpreted as  Lyapunov functions for the vector process $\mu (\cdot)$ of market weights; that is, via the property that $G (\mu (\cdot))$ is a supermartingale under an appropriate change of   measure.  This point of view unifies, generalizes, and simplifies several existing results, and allows the formulation of   conditions under which it is possible to outperform the market portfolio over appropriate time-horizons.  From a   
probabilistic point of view, the present paper yields results concerning the interplay of stochastic discount factors and concave transformations of semimartingales on compact domains. 
\end{abstract}

\smallskip
\noindent{\it Keywords and Phrases:} Trading strategies, functional generation, outperformance, relative arbitrage, regular and Lyapunov functions, concavity, semimartingale property, deflators.

\smallskip
\noindent{\it AMS 2000 Subject Classifications:} 60G44, 60H05,   60H30,  91G10,  93D30.

\input amssym.def
\input amssym

 \smallskip

\section{Introduction}

 Back in 1999, E.R.~Fernholz introduced   a construction that was both    remarkable and    remarkably easy to  establish. He showed that for a certain class of so-called ``functionally-generated'' portfolios, it is possible to express the wealth these portfolios generate, discounted by (that is, denominated in terms of) the total market capitalization, solely in terms of the individual companies'  {\it market weights} -- and to do so in a pathwise manner, that  {\it does not involve stochastic integration.}  This fact can be proved by a somewhat determined application of {It\^o}'s rule.  
 Once the result is known,   its proof becomes a  moderate exercise in stochastic calculus.

 The discovery paved the way for finding simple and very general structural conditions on  {\it large}  equity markets -- that involve more than   one stock, and typically thousands --  under which it is   possible strictly to outperform the market portfolio.  Put a little differently: conditions under which strong  relative  arbitrage with respect  to the market portfolio is possible, at least over sufficiently long time-horizons.  \citet{F_generating, Fernholz:2001, Fe} showed also  how to implement this strong  relative  arbitrage, or ``outperformance,'' using   portfolios that can be   constructed solely in terms of observable quantities,  and  without any need for estimation  or optimization.     \cite{Pal:Wong:Geometry} related    functional generation   to optimal transport in discrete time.
 
 Although well-known, celebrated, and quite easy to prove, Fernholz's construction has been viewed over the past    15  years as somewhat ``mysterious.''  In this paper   we hope  to help make the result  a bit more celebrated and a bit  less mysterious, via an interpretation of portfolio-generating functions $G$ as {Lyapunov} functions for the vector process $\mu (\cdot)$ of  
     relative market weights. Namely,  via the property that $G (\mu (\cdot)) $ is a supermartingale under an appropriate change of   measure; see Remark~\ref{REM: 2a} for elaboration. We  generalize this functional generation from portfolios to trading strategies, as well as to situations where some,   but not all, of the market weights can vanish;   along  the way we simplify the underlying arguments considerably, and answer an old question of  \cite{Fe}, Problem~4.2.3. Conditions for strong  outperformance of the market over appropriate time horizons become extremely simple via this interpretation, as do the strategies that implement  such outperformance and the accompanying proofs that establish  such results; see Theorems~\ref{thm: 2} and \ref{thm: 3}. 
   
   We have cast all our results in the framework of continuous semimartingales for the market weights; this seems to us a very good compromise between generality on the one hand, and conciseness, unity  and readability on the other. The reader will easily decide which of the results can be extended to general semimartingales, and which cannot. 
   
  Here is an outline of the paper. Section~\ref{sec: 2} presents the market model and recalls the financial concepts of trading strategies, outperformance, and deflators.  Section~\ref{sec: 4} then introduces the notions of regular and Lyapunov functions. Section~\ref{sec: 5} discusses how such functions generate trading strategies, and Section~\ref{sec: 6} uses these observations to formulate conditions that guarantee trading strategies which outperform the market over sufficiently long time horizons. Section~\ref{sec: 7} contains several relevant examples for regular and Lyapunov functions and the corresponding generated strategies. Section~\ref{sec: 8} proves that concave functions satisfying certain additional assumptions are indeed Lyapunov  and provides counterexamples if those additional assumptions are not satisfied. Finally, Section~\ref{sec: 9} concludes.

\section{The setup} 
 \label{sec: 2}

\subsection{Market model}
On a  given probability space $(\Omega, \sigalgebra{F}, \P)$ endowed  with a right-continuous filtration $ \filtration{F} = ( \sigalgebra{F} (t) )_{ t \geq 0}$ that   satisfies $\sigalgebra{F}(0) = \{ \emptyset, \Omega \}$ mod.~$\P$, we consider a vector process
$
{S} ( \cdot)= \transpose{( S_1 ( \cdot), \cdots, S_d ( \cdot) )}
$
of  continuous, non-negative semimartingales with $ S_1 (0)>0, \cdots , S_d (0)>0$ and 
\begin{equation}
\label{eq: total}
\Sigma (t) := S_1 ( t)+  \cdots + S_d ( t)>0, \qquad  t \geq 0. 
\end{equation}
 We interpret these processes as the capitalizations of  a fixed number $d \ge 2$   of companies in an equity market. An individual company's capitalization $ S_i (\cdot)$ is allowed to vanish; but the total capitalization $ \Sigma (\cdot)$ of the equity market is not.     Throughout this paper  we   study trading strategies that only invest in these $d$ assets, and abstain from introducing a money market explicitly: the financial market   of available investment opportunities is represented here by the $d$--dimensional continuous semimartingale $ S(\cdot)$. 

Having introduced these quantities, we now define  the vector process $ \mu (\cdot) = \transpose{( \mu_1 (\cdot), \cdots, \mu_d (\cdot) )}$ that consists of the various companies' relative {\it market weights}
\begin{equation}
\label{eq: mw}
\mu_i (t) := \frac{S_i (t)}{ \Sigma (t)}  = \frac{S_i (t)}{ S_1 (t) + \cdots + S_d (t)}, \qquad  t \geq 0 
\end{equation}
for each $ i=1, \cdots, d$. These processes are continuous,  non-negative semimartingales in their own right; each of them takes values in the unit interval [0,1] and they satisfy $ \mu_1 (\cdot)+ \cdots+ \mu_d (\cdot) \equiv 1.$ In other words, the vector process $ \mu (\cdot)$ takes values in the lateral face ${\bm \Delta}^d$ of the unit simplex in $ \R^d.$ We are using throughout   the notation 
\begin{equation}
\label{eq: Delta}
{\bm \Delta}^d := \bigg\{ \transpose{\big(x_1, \cdots, x_d \big)} \in [0, 1]^d:\, \sum_{i=1}^d x_i =1 \bigg\} , \qquad {\bm \Delta}^d_+ := {\bm \Delta}^d \cap (0, 1)^d 
\end{equation}
and note that, by assumption, $ \mu (0) \in {\bm \Delta}^d_+.$   

An important special case of the above setup arises, when each semimartingale $ S_i ( \cdot)$ is strictly positive; equivalently, when the process $ \mu (\cdot)$ takes values in $ {\bm \Delta}^d_+, $ that is,   
\begin{equation}
\label{eq: interiors}
\P\big( \mu (t) \in   {\bm \Delta}^d_+, \,\, \forall \,\, t \geq 0 \big) =1.
\end{equation}

\subsection{Trading strategies}
Let $X ( \cdot) =\transpose{\big( X_1 ( \cdot) , \cdots, X_d ( \cdot) \big)}$ denote a generic $[0,\infty)^d$--valued continuous semimartingale.  For the purposes of this  section, $X ( \cdot)$ will stand for either the vector process $S(\cdot)$ of  capitalizations, or for  the vector process $\mu(\cdot)$ of market weights. We consider   a predictable process 
$
\vartheta ( \cdot) = \transpose{\big( \vartheta_1 ( \cdot) , \cdots, \vartheta_d ( \cdot) \big)}
$
with values in $ \R^d$, and interpret  $ \vartheta_i ( t)$ as the number of shares held at time $t \geq 0$ in the stock of company $ i=1, \cdots, d$.  Then  the total {\it value}, or ``wealth,''  of this investment
 in a market whose price processes are given by the vector process $X(\cdot)\,,$ is
\begin{equation} \label{eq:160102.1}
V^\vartheta (\cdot \,; X):= \sum_{i=1}^d  \vartheta_i ( \cdot) X_i (\cdot).
\end{equation}

\begin{defn} [Trading strategies]
\label{def: TS}
Suppose that the $ \R^d$--valued, predictable process $ \vartheta (\cdot)$ is integrable with respect to the continuous semimartingale $X( \cdot) $; and write $ \vartheta (\cdot) \in \mathscr L (X)$ to express this. We shall say that such $ \vartheta (\cdot) \in \mathscr L (X)$ is a  {\it trading strategy} with respect to $X( \cdot)$, if the so-called ``self-financibility'' condition
\begin{equation}
\label{eq: SF}
V^\vartheta (\cdot \,; X)- V^\vartheta (0; X) = \int_0^\cdot \big \langle \vartheta (t), \mathrm{d} X(t) \big \rangle
\end{equation}
is satisfied. We shall denote by $   \mathscr T (X)$ the collection of all such trading strategies.  \qed
\end{defn}

\begin{rem} [On notation and interpretation] Here and in what follows,  we use for any fixed $T\geq 0$  the notation on the right-hand side of \eqref{eq: SF}, namely 
\begin{equation*}
\int_0^T \big \langle \vartheta (t), \mathrm{d}X(t) \big \rangle = \int_0^T \sum_{i=1}^d  \vartheta_i ( t) \mathrm{d} X_i (t),
\end{equation*}
as a short-hand for vector stochastic integration. This quantity gives the ``gains-from-trade'' realized over the interval $[0,T]$ (gains, if it is positive; losses, if it is negative). The self-financibility requirement of \eqref{eq: SF} posits that these ``gains'' account for the entire change in the value generated by the trading strategy $ \vartheta (\cdot)$ between the start $ t=0$ and the end $ t=T$ of the time-interval $[0,T]$: there is no infusion  of funds, and neither  are there transaction or other fees.  \qed
\end{rem}

The following result can be proved via a somewhat determined application of {It\^o}'s rule. It formalizes the intuitive idea that the concept of trading  strategy should not depend on the manner in which prices or capitalizations are quoted. We refer to Proposition~1 in \cite{Ger} for a   proof.

\begin{prop}[Change of num\'eraire]
\label{pr: ChoN}
An $ \R^d$--valued process $ \vartheta (\cdot)=\transpose{\big(\vartheta_1 ( \cdot) , \cdots, \vartheta_d ( \cdot) \big)}$ is a  trading strategy  with respect to the $ \R^d$--valued semimartingale $S( \cdot) $, if and only if it is a trading  strategy  with respect to the $ \R^d$--valued semimartingale $\mu(\cdot)$ given in \eqref{eq: mw}.
In particular,   $   \mathscr T (\mathcal{S})=  \mathscr T (\mathcal{\mu});$  and in this case, we have 
$
V^\vartheta (\cdot \, ; S)= \Sigma (\cdot) V^\vartheta (\cdot  \,; \mu)$. 
\end{prop}

Suppose we are given  an element $ \vartheta (\cdot) = \transpose{( \vartheta_1 (\cdot), \cdots, \vartheta_d (\cdot))}$ in  the space $ \mathscr L (\mu) $ of   predictable 
  processes  
  which are integrable with respect to the continuous vector semimartingale $\mu ( \cdot)= \transpose{\big( \mu_1 (\cdot), \cdots, \mu_d (\cdot) \big)} $ of \eqref{eq: mw}.  Let us   consider  the  quantity
\begin{equation}
\label{eq: cue}
Q^\vartheta (T  ; \mu ) \,:= \,V^\vartheta (T ; \mu ) - V^\vartheta (0 ; \mu )-  \int_0^T \big \langle \vartheta (t), \mathrm{d} \mu (t) \big \rangle, \qquad T \geq 0,
\end{equation}
which measures the ``defect of self-financibility'' of this process $ \vartheta (\cdot)$ relative to $ \mu (\cdot)$ over the time-horizon $ [0,T]$.  
If $ Q^\vartheta (\cdot\, ; \mu ) \equiv 0$ fails, the process $ \vartheta (\cdot) \in \mathscr L (\mu)$ is not a trading strategy with respect to $ \mu (\cdot).$  How do we modify it then, in order to turn it into a  trading  strategy? Our next result describes a way, which essentially adjusts each  component of $ \vartheta (\cdot)  $ by the defect of self-financibility. 

\begin{prop}[From integrands to trading strategies] \label{pr: ITS}
For a given  process $ \vartheta (\cdot) \in \mathscr L (\mu)$,  a given real constant $\bm C \in \R,$  and with the   notation of \eqref{eq: cue}, we introduce the processes
\begin{equation}
\label{eq: phi}
\varphi_i ( t) \,:=  \,
\vartheta_i ( t) - Q^\vartheta ( t \, ; \mu ) + \bm C, \qquad  i=1, \cdots, d\,, \quad   t \ge 0\,. 
\end{equation}
 The resulting $ \R^{n}$--valued, predictable process $ \varphi (\cdot) = \transpose{\big( \varphi_1 (\cdot), \cdots, \varphi_d (\cdot) \big)}$ is then a  trading  strategy with respect to the vector process $ \mu (\cdot)$ of market weights;         to wit, $ \varphi (\cdot) \in \mathscr T (\mu)$. 
 Moreover, the  value process $V^\varphi (\cdot\, ; \mu ) = \sum_{i=1}^d  \varphi_i ( \cdot) \mu_i (\cdot)$ of this trading strategy satisfies
\begin{equation}
\label{eq: valphi}
V^\varphi (\cdot\, ; \mu ) = V^\vartheta (0\, ; \mu ) + \bm C + \int_0^{ \cdot} \big \langle \vartheta (t), \mathrm{d} \mu (t) \big \rangle = V^\varphi (0\, ; \mu )   + \int_0^{ \cdot} \big \langle \varphi (t), \mathrm{d} \mu (t) \big \rangle .
\end{equation}
\end{prop} 

\begin{proof}
Consider the vector process $\,\widetilde \vartheta(\cdot) =  \transpose{\big( \widetilde \vartheta_1 (\cdot), \cdots, \widetilde \vartheta_d (\cdot) \big)}$  with   components   $\,\widetilde \vartheta_i(\cdot) = \bm C -Q^\vartheta (\cdot \, ; \mu )$ for each   $i = 1, \cdots, d\,$.  Then $\widetilde \vartheta(\cdot)$ is predictable, since $V^\vartheta (\cdot \, ; \mu )$ and $\int_0^\cdot \big \langle \vartheta (t), \mathrm{d} \mu (t) \big \rangle$ are.  Moreover, Lemma~4.13 in \cite{Shiryaev_vector} yields $\widetilde \vartheta(\cdot)  \in 
\mathscr L (\mu)$; thus, we have  also $ \varphi (\cdot) = \vartheta (\cdot) +\widetilde{ \vartheta} (\cdot) \in \mathscr L (\mu)$.  Furthermore, 
$$
\int_0^\cdot \big \langle \widetilde \vartheta (t), \mathrm{d} \mu (t) \big \rangle \equiv 0
$$
holds thanks to $ \sum_{i=1}^d \mu_i (\cdot) \equiv 1$, and therefore so does 
$$ \int_0^\cdot \big \langle  \vartheta (t), \mathrm{d} \mu (t) \big \rangle = \int_0^\cdot \big \langle  \varphi (t), \mathrm{d} \mu (t) \big \rangle.$$
Since we also have  $ \varphi_i (0) = \vartheta_i (0) + \bm C$ for each $i=1, \cdots, d,$ hence $ V^\varphi (0 ; \mu  ) =V^\vartheta (0 ; \mu ) + \bm C$, we obtain \eqref{eq: valphi}. This yields that $\varphi(\cdot)$ is indeed a trading strategy, and concludes the proof.
\end{proof}

\subsection{Outperforming the market} 
 \label{sec: 3b}

Let us fix a real number $ T > 0.$ We say that a given     trading  strategy $\varphi (\cdot) \in      \mathscr T (S)$    {\it outperforms the market over the time-horizon $[0,T]\,,$} if   we have
 \begin{equation}
 \label{eq: A1}
 V^\varphi (t; S)  \ge  0, \,\,\, \forall \,\, t \in [0,T]; \qquad V^\varphi (0; S) =  \Sigma (0)
 \end{equation} 
 in the notation of \eqref{eq: total}, along with 
\begin{equation}
\label{eq: A2}
\P \left( 
V^\varphi (T; S)  \ge  \Sigma (T)  \right) =1; \qquad 
\P \left( 
V^\varphi (T; S) >  \Sigma (T)  \right) >0.
\end{equation}

Whenever a given $ \varphi (\cdot) \in      \mathscr T (S )$ satisfies these conditions, and if in fact the second probability in \eqref{eq: A2} is not just positive but actually equal to 1, that is, if 
\begin{equation}
\label{eq: A3}
\P \left( 
V^\varphi (T; S) >  \Sigma (T)  \right) =1 
\end{equation}
holds, then we say that this $ \varphi (\cdot)  $   {\it outperforms the market strongly}    over   $[0,T]$.    

\begin{rem}[Change of num\'eraire]
  \label{rem: 1}
It follows from Proposition~\ref{pr: ChoN} that the above requirements  \eqref{eq: A1}--\eqref{eq: A3} 
can be cast, respectively,   as 
\begin{align}
V^\varphi (t; \mu)  \ge  0, \,\,\, \forall \,\, t \in [0,T]; \qquad V^\varphi (0; \mu) =  1;
\label{eq:160130.1}\\
\P \left( 
V^\varphi (T; \mu)  \ge  1  \right) =1; \qquad 
\P \left( 
V^\varphi (T; \mu) >  1  \right) >0  \nonumber
\end{align}
and $ \P ( 
V^\varphi (T; \mu) >  1  ) =1$.  \qed
\end{rem}

We remark that the concept of (strong) outperformance is often also called (strong) relative arbitrage in the literature; see, for example, \cite{FKK} and \cite{FK_survey}.

\subsection{Deflators} 
 \label{sec: 2a}

For some of our results we shall need the notion of {\it deflator} for the vector process $ \mu (\cdot)$ of market weights in \eqref{eq: mw}. This is any continuous,   strictly positive and adapted process $ Z(\cdot)$ with $ Z(0)=1$ for which 
\begin{equation}
\label{eq: 5a}
\text{all products}\quad  Z(\cdot) \, \mu_i (\cdot),\,\,\,\,i=1, \cdots, d\,\quad \text{are local martingales;}
\end{equation}
 thus  $ Z(\cdot)$ is also a local martingale itself. An apparently stronger condition  is that the product
\begin{equation}
\label{eq: 5}
Z(\cdot)\int_0^{ \cdot} \big \langle \vartheta (t), \mathrm{d} \mu (t) \big \rangle\quad \text{is a local martingale, for every}\,\,\, \vartheta (\cdot) \in \mathscr L (\mu).
\end{equation}

\begin{prop}[Equivalence of conditions] \label{R:160110}
The conditions in \eqref{eq: 5a} and \eqref{eq: 5} are equivalent. 
\end{prop}
\begin{proof} Let us suppose \eqref{eq: 5a} holds; then $ Z(\cdot)$ is a 
local martingale, so there exists a nondecreasing sequence $ (\tau_n)_{n \in \N}$ of stopping times with $ \lim_{n \uparrow \infty}  \tau_n = \infty$ and the property that  $ Z( \cdot \wedge \tau_n)$ is a uniformly integrable martingale for each $n\in \N$.  
The recipe $ \Qu_n (A)= \E^\P [ Z(\tau_n) \mathbf{ 1}_A]$ for each $A\in \sigalgebra{F} (\tau_n)$ defines a probability measure on $ \sigalgebra{F} (\tau_n)$, under which   $ \mu_i ( \cdot \wedge \tau_n)$ is a   martingale,  for each  $i=1, \cdots, d$ and $n \in \N$.   
However, the ``stopped'' version    $\int_0^{ \cdot \wedge \tau_n} \big \langle \vartheta (t), \mathrm{d} \mu (t) \big \rangle$ of the  stochastic integral as in \eqref{eq: 5} is  then also a $ \Qu_n$--local martingale for each $n \in \N$; therefore each product  $Z( \cdot \wedge \tau_n)\int_0^{ \cdot \wedge \tau_n} \big \langle \vartheta (t), \mathrm{d} \mu (t) \big \rangle$  is   a $ \P$--local martingale, and the property of \eqref{eq: 5} follows. The  reverse  implication   is trivial. 
\end{proof}
 
 \begin{rem}[Equivalent martingale measure for market weights]
 \label{rem: 5}
If a deflator $Z(\cdot)$ exists and is a martingale, then for any real number $ T>0$ we can define a probability measure on $ \sigalgebra{F} (T) $ by $\Qu_T (A)= \E^\P [ Z(T)  \mathbf{ 1}_A]\,,$   $\,A\in \sigalgebra{F} (T)\,.$   Under this measure the market weights $ \mu_i ( \cdot   \wedge T), \,\, i=1, \cdots, d$ are local martingales; thus actual martingales, as they take values in [0,1]. \qed
\end{rem}

Now let us introduce the stopping times
\begin{equation}
\label{eq: 6}
\mathscr D \,:= \,\mathscr D_1 \wedge \cdots \wedge \mathscr D_d, \qquad \mathscr D_i\,: =\, \inf \big\{ t \ge 0 :\, \mu_i (t) = 0\, \big\}.
\end{equation}
Whenever a deflator  for the vector  $ \mu (\cdot)$ of market weights exists,  each continuous process $Z(\cdot)\,\mu_i (\cdot),$ being non-negative and a local martingale, is   a supermartingale. From this, and from the strict positivity of $ Z(\cdot),$ we see that then
\begin{equation*}
\mu_i \big( \mathscr D_i +u \big) = 0 \,\, \text{holds for all $u \geq 0,$ on the event $\big\{  \mathscr D_i < \infty \big\}$.}
\end{equation*}

Thus, the vector process $ \mu (\cdot)$ of market weights starts life at a point  $ \mu (0) \in {\bm \Delta}^d_+.$ It may  then -- that is, when a deflator exists -- begin a ``descent'' into simplices of successively lower dimensions, possibly all the way   up until the time the entire market capitalization     concentrates in just one company, that is 
\begin{equation}
\label{eq: 6*}
\mathscr D_\star \,:= \,\inf \big\{  t \ge 0 :\, \mu_i (t) = 1 \quad \text{for some $i=1, \cdots, d$} \big\}.
\end{equation}

\section{Regular and Lyapunov functions} 
 \label{sec: 4}

For a generic $d$-dimensional semimartingale $X(\cdot)$ we write $\,\mathbf{supp\,}( X)\,$ to denote the support of $X(\cdot)$, that is, the smallest closed set $\mathfrak S \subset \R^d$ such that
 $$\P(X(t) \in \mathfrak S, \,\, \forall \,\,t \geq 0)  = 1.$$ 
 In case of the vector process $\mu(\cdot)$ of market weights  in \eqref{eq: mw} we always have $\mathbf{supp\,} ( \mu )\subset \bm \Delta^d$ for the lateral face ${\bm \Delta}^d$ of the unit simplex, defined in \eqref{eq: Delta}.

 \begin{defn}[Regular functions]
  \label{def: reg}
We say that a continuous function $ G : \mathbf{supp\,} (X )\rightarrow \R$   is {\it regular} for the $d$-dimensional semimartingale $X(\cdot)$ if 
\begin{enumerate}[label={\rm(\roman*)},ref={\rm(\roman*)}]
\item there exists a measurable function $ DG = \transpose{\big( D_1 G, \cdots ,  D_d G\big)} :  \mathbf{supp\,} (X) \rightarrow \R^d$ such that the process $ {\bm \vartheta}   ( \cdot) =\transpose{\big( {\bm \vartheta}_1 ( \cdot), \cdots,  {\bm \vartheta}_d ( \cdot) \big)} $ with components 
\begin{equation}
\label{eq: thetaG}
{\bm \vartheta}_i ( t) \,:= 
D_i G \big( X ( t) \big), \qquad i=1, \cdots, d\,,  \quad   t \ge 0
\end{equation}
is in $ \mathscr L (X);$ and  
\item     the  continuous,   adapted process
\begin{equation}
\label{eq: deco}
\Gamma^G (T)\,:=\,  
  G \big( X (0) \big) -  G \big( X (T) \big) +  \int_0^{ T} \big \langle \bm  \vartheta (t), \mathrm{d} X (t) \big \rangle\,, \qquad   T \ge 0
\end{equation}
has finite variation on compact intervals. \qed
\end{enumerate}
\end{defn}

 \begin{defn} [Lyapunov functions] 
  \label{def: Lyap}
We say that a regular function $G$ as in Definition \ref{def: reg}   is {\it a Lyapunov function} for the   $d$-dimensional semimartingale $X(\cdot) $  if, for some function  $DG$ as in Definition~\ref{def: reg}, the finite-variation  process $ \Gamma^G (\cdot) $ of   \eqref{eq: deco} is actually non-decreasing.     \qed
\end{defn}

\begin{rem}[Supermartingale properties]
 \label{REM: 2a}
Let us  suppose that a probability measure $ \Qu$ exists, under which the market weights $ \mu_1 (\cdot), \cdots,  \mu_d (\cdot)$ are (local) martingales. Then, for any given regular function $G : \mathbf{supp\,} (\mu )\rightarrow \R$, it is seen from  \eqref{eq: deco} that  the continuous process 
\begin{equation}
\label{eq: decor}
 G \big( \mu (\cdot) \big) + \Gamma^G(\cdot)  = G \big( \mu (0) \big)  + \int_0^\cdot \sum_{i=1}^d D_i  G \big(  \mu ( t) \big)  \mathrm{d}  \mu_i (t) 
 \end{equation}
  is  a $\Qu$--local martingale, provided that $G$ is regular for $\mu(\cdot)$. If, furthermore, this $G$ is actually a 
  {Lyapunov} function   for $\mu(\cdot),$ then it follows  that the process $G \big( \mu (\cdot) \big)$   is  a $\Qu$--local supermartingale -- thus in fact a  $\Qu$--supermartingale, as it is bounded from below due to the continuity of $G$.  

       A bit   more generally,    let us assume now that there exists a deflator $ Z(\cdot)$    for the market weight process $\mu(\cdot)$.  Then Proposition~\ref{R:160110} yields that the product  $ Z(\cdot) \int_0^\cdot \sum_{i=1}^d D_i  G \big(  \mu ( t) \big)  \mathrm{d}  \mu_i (t)  $ is   a $\P$--local martingale.  
       If  now   $G $ is   a   
       {Lyapunov} function   for $\mu(\cdot)$,    integration by parts shows that  the process  $$  Z(\cdot) G \big( \mu (\cdot) \big) =  G(\mu(0)) + Z(\cdot)  \int_0^\cdot \sum_{i=1}^d D_i  G \big(  \mu ( t) \big)  \mathrm{d}  \mu_i (t)  - \int_0^{ \cdot} \Gamma^G (t) \mathrm{d} Z(t) - \int_0^{ \cdot} Z(t)  \mathrm{d}  \Gamma^G (t)$$ is     a  $\P$--local supermartingale, thus also a  $\P$--supermartingale as it is bounded from below.     
   \qed
   \end{rem}

The process $\Gamma^G(\cdot)$ in \eqref{eq: deco} might depend on the choice of $DG$. For  example, consider the situation when each component of $\mu(\cdot)$ is of first finite variation but not constant. Then it is easy to see that different choices of $DG$ lead to different processes $\Gamma^G(\cdot)$ in \eqref{eq: deco}. However, if a deflator for $\mu(\cdot)$ exists then we get the following uniqueness result.

\begin{prop} [Uniqueness in \eqref{eq: deco}]
 \label{pr: unique}
If a function $ G :  \mathbf{supp\,} (\mu ) \rightarrow \R$    is regular  for the vector process  
$ \mu (\cdot) = \transpose{\big(   \mu_1 (\cdot), \cdots, \mu_d (\cdot) \big)}$ of market weights, and if a deflator for the process $ \mu (\cdot)$ exists, then the continuous, adapted, finite-variation  process $ \Gamma^G (\cdot)$ of \eqref{eq: deco} does not depend on the choice of $DG$.
 \end{prop}
\begin{proof}
Suppose that there exist   a deflator $Z(\cdot)$ for the vector process $\mu (\cdot)$ of market weights; as well as  two functions $DG$, $\widetilde{DG}$ as in Definition \ref{def: reg}, with corresponding processes ${\bm \vartheta} (\cdot)$, $\widetilde{{\bm \vartheta}} (\cdot)$ in \eqref{eq: thetaG} and $ \Gamma^G (\cdot)$, $\widetilde{ \Gamma}^G (\cdot)$ in \eqref{eq: deco}.  We need to show $ \Gamma^G (\cdot) = \widetilde{ \Gamma}^G (\cdot),$ or equivalently 
$$
\Upsilon (\cdot) := \int_0^\cdot \big \langle {\bm \phi} (t), \mathrm{d} \mu (t) \big \rangle \equiv 0\,, \qquad \text{where} \quad {\bm \phi} (\cdot) := {\bm \vartheta} (\cdot)-\widetilde{{\bm \vartheta}} (\cdot)\,.
$$
Now, on the strength of \eqref{eq: deco}, this continuous process $\Upsilon(\cdot)$ is of finite variation on compact intervals, so the product rule gives
$$
\int_0^\cdot Z (t)\, \mathrm{d} \Upsilon (t) = Z(\cdot)\, \Upsilon(\cdot) - \int_0^\cdot \Upsilon (t)\, \mathrm{d} Z (t)  \,= Z(\cdot)  \int_0^\cdot \big \langle {\bm \phi} (t), \mathrm{d} \mu (t) \big \rangle - \int_0^\cdot \Upsilon (t)\, \mathrm{d} Z (t) .
$$
As a consequence of Proposition \ref{R:160110}, the process on the right-hand side is a local martingale; on the other hand, the process $ \int_0^\cdot Z (t)\, \mathrm{d} \Upsilon (t)$ is   continuous and of finite variation on compact intervals, and thus identically equal to zero. The strict positivity of $Z(\cdot)$ gives now $ \Upsilon(\cdot) \equiv 0$.  
\end{proof}

	

\subsection{Sufficient conditions for a function to be regular or Lyapunov} 

\begin{example}[The smooth case]  \label{ex: smooth}
Suppose that a given continuous function $ G :\mathbf{supp\,} (\mu ) \rightarrow \R$ can be extended to a twice continuously differentiable function on some open set $\,\mathcal U \subset \R^d$ with 
\begin{equation*}
\P \big( \mu (t) \in   \mathcal U, \,\, \forall \,\, t \geq 0 \big) =1.
\end{equation*} 
  Elementary stochastic calculus expresses then the process of \eqref{eq: deco} as  
\begin{equation}
\label{eq: Gamma}
  \Gamma^G (\cdot)= -
\frac{1}{2}  \sum_{i=1}^d \sum_{j=1}^d \int_0^{ \cdot}    D^2_{ij}
G \big(\mu (t)  \big)   \,  \mathrm{d} \big \langle \mu_i, \mu_j  \big \rangle (t)
     \end{equation}
with the notation  $ D_i G = \partial G / \partial x_i$,   $ D_{ij}^2 G = \partial^2 G / (\partial x_i \partial x_j)$. (See Propositions~4 and 6 in \citet{Bouleau:1984} for   slight generalizations of this result.) Therefore, such a function $G$ is regular; if it is also concave, then the process $ \Gamma^G (\cdot)$ in \eqref{eq: Gamma} is   non-decreasing,  and $G$ becomes    a {Lyapunov} function.  \qed
\end{example}

Quite a bit more generally, we have the following results. 

\begin{thm}[The concave case]
\label{thm: 1}
A  given continuous function $ G :  \mathbf{supp\,} (\mu ) \rightarrow \R$   is a {Lyapunov} function for the vector process $ \mu (\cdot)$ of market weights if one of the following conditions holds: 
\begin{enumerate}[label={\rm(\roman*)},ref={\rm(\roman*)}]
\item\label{thm: 1i}  $G$  can be extended to a continuous, concave function on ${\bm \Delta}^d_+\,,$  and  \eqref{eq: interiors} holds. 
\item\label{thm: 1ii} $G$ can be extended to a continuous, concave function on the set 
\begin{equation}
\label{eq: Delta_e}
{\bm \Delta}^d_e \,:= \, \bigg\{ \transpose{  \big(x_1, \cdots, x_d \big)} \in \R^d:\, \sum_{i=1}^d x_i =1 \bigg\}. \end{equation}
\item\label{thm: 1iii} $G$ can be extended to a  continuous, concave   function on the set  $\bm \Delta^d$ of \eqref{eq: Delta}, and there exists a deflator for the vector process $ \mu (\cdot) = \transpose{\big(   \mu_1 (\cdot), \cdots, \mu_d (\cdot) \big)}$ of market weights.    
\end{enumerate}
\end{thm}

  We refer to Section~\ref{sec: 8} for a review of some basic notions from convexity, and for the proof of Theorem~\ref{thm: 1}.  The existence of a deflator is essential for the sufficiency in Theorem~\ref{thm: 1}\ref{thm: 1iii} (that is, whenever the market-weight process $\mu(\cdot)$ is ``allowed to hit a boundary"), as illustrated by  Example~\ref{ex:Counterexample1}  below.

\subsection{Rank-based regular and Lyapunov functions} 
Let us introduce the ``rank operator'' $ \mathfrak R,$ namely, the mapping $ {\bm \Delta}^d \ni (x_1, \cdots, x_d ) \mapsto \mathfrak R (x_1, \cdots, x_d )= (x_{(1)}, \cdots, x_{(d)} )\in \mathbb{W}^{d} \,,$ where 
\begin{align} \label{eq: Wn}
\mathbb{W}^{d}\,: =\,  \Big\{  \transpose{\big(x_1, \cdots, x_d \big)} \in {\bm \Delta}^{d} :\, 1 \geq x_1 \ge x_2 \ge \cdots \ge x_{d-1} \ge x_d \geq 0 \Big\}.
\end{align}
We denote by 
$$
\max_{i = 1, \cdots, d} \, x_i \,= \, x_{(1)}\ge  x_{(2)}\ge \cdots \ge  x_{(d-1)} \ge x_{(d)}= \min_{i = 1, \cdots, d} \,x_i 
$$
  the descending order statistics of the components of the vector $x= (x_1, \cdots, x_d)'\,,$ constructed with a clear rule for breaking ties (say, the lexicographic rule that always favors the smallest ``index'' $i=1, \cdots, d$); moreover,  for each $x \in \bm \Delta^d$, we denote by 
\begin{align} \label{eq: Nell}
 N_\ell(x) \,: =\,  \sum_{i=1}^d \,  \mathbf{ 1}_{ \,  x_{(\ell)} = x_i   }
\end{align}
  the number of components of the vector $x= (x_1, \cdots, x_d)'\,,$ that coalesce   in the given  rank $\,\ell = 1, \cdots, d$. 
  
  Finally, we introduce the process  of   market weights ranked in descending order, namely 
\begin{align}	
\label{eq: R}
{\bm \mu} ( t) = \mathfrak{ R} (\mu ( t)) = \transpose{\big( \mu_{(1)} ( t), \cdots ,  \mu_{(d)} ( t) \big)}\,,\qquad   t \ge 0\,.
\end{align}
 We note that  ${\bm \mu} (\cdot)$ can be interpreted again as a market model. However, this rank-based model may fail to admit a deflator, even when the original vector process of market weights  $\mu (\cdot)$ admits one. This is due to the appearance,  in the dynamics for  $ {\bm \mu} (\cdot)$, of local time terms, which correspond  to the reflections whenever two or more components of the original process $\mu(\cdot)$ collide with each other; see, for instance, \eqref{eq:160128.1} below.

\begin{thm}[The concave case, continued]
\label{thm: 1b}
	Consider a  function ${\bm G} : \mathbf{supp\,} ( \bm \mu ) \rightarrow \R$. Then ${\bm G}$ is a Lyapunov function for the vector process $ {\bm \mu} (\cdot)$ in \eqref{eq: R}, if one of the following two conditions holds.
	\begin{enumerate}[label={\rm(\roman*)},ref={\rm(\roman*)}]
\item\label{thm: 1bi}  $\bm G$ can be extended to a continuous, concave function on $\,{\bm \Delta}^d_+\,,$ and  \eqref{eq: interiors} holds; or 
\item\label{thm: 1biii} $\bm G$ can be extended to a continuous, concave function on $\,{\bm \Delta}^d_e\,,$  given in \eqref{eq: Delta_e}.
\end{enumerate}
Under any of these two conditions, the composition $\,G = {\bm G} \circ \mathfrak R\,$ is a regular function for the vector process $\mu(\cdot)$.      
More generally, if $\bm G$ is a regular function for  $ {\bm \mu} (\cdot)$, then  $\,G = {\bm G} \circ \mathfrak R\,$ is a regular function for  $\mu(\cdot)$. 
\end{thm}
We refer again to Section~\ref{sec: 8} for  the proof of Theorem~\ref{thm: 1b}.  A simple modification of Example~\ref{ex:Counterexample1} illustrates that a function $\bm G$ can be  concave and continuous on $\mathbb{W}^{d}$ without being regular for $\mu(\cdot)$.  Indeed, this can happen even when  a deflator for $\mu(\cdot)$ exists, as Example~\ref{ex:Counterexample2} illustrates.

\begin{example}[The smooth case, continued]  \label{ex: smooth2}
 Example~\ref{ex: smooth} has an  equivalent formulation for the rank-based case.
Assume again that the function $\bm G :\mathbf{supp\,} ( \bm \mu )  \rightarrow \R$ can be extended to a twice continuously differentiable function on some open set $\, \mathcal U \subset \R^d\,$ with 
\begin{equation*}
\P \big(\bm \mu (t) \in   \mathcal U, \,\, \forall \,\, t \geq 0 \big) =1.
\end{equation*}
Then $\bm G$ is regular for $\bm \mu(\cdot)$. Indeed,   as in Example~\ref{ex: smooth}, applying It\^o's formula yields
\begin{align*}
	\bm G(\bm \mu(\cdot)) = \bm G(\bm \mu(0)) + \int_0^\cdot \sum_{\ell=1}^d D_\ell \bm G \big( \bm \mu ( t) \big)  \mathrm{d} \bm \mu_\ell (t) + \frac{1}{2}  \sum_{k=1}^d \sum_{\ell=1}^d \int_0^{ \cdot}    D^2_{k \ell}
\bm G \big(\bm \mu (t)  \big)     \mathrm{d} \big \langle\bm  \mu_k,\bm \mu_\ell \big \rangle (t)
\end{align*}
with   $ D_\ell \bm G = \partial \bm G / \partial x_\ell$,   $ D_{k\ell}^2 \bm G = \partial^2 \bm G / (\partial x_k \partial x_\ell)$, and the regularity of $\bm G$    for $\bm \mu(\cdot)$ follows.

\smallskip
Next, let $ \Lambda^{(k, \ell )} (\cdot)$ denote the local time process  of the continuous  semimartingale  
$\mu_{(k)}(\cdot) - \mu_{(\ell)}(\cdot) \ge 0 $   at the origin, for   $1 \le k < \ell \le d$.   Then with   the notation of \eqref{eq: Nell}, 
Theorem~2.3  in \citet{Banner:Ghomrasni} yields the semimartingale representation   for the ranked market weights 
\begin{align} 
	\bm \mu_\ell(\cdot) &= \bm \mu_\ell(0) + \int_0^\cdot \sum_{i=1}^d \frac{1}{N_\ell(\mu(t))}   \1_{\{\mu_{(\ell)}(t)  = \mu_i(t)\}} \mathrm{d} \mu_i(t) + \sum_{k=\ell+1}^d \int_0^\cdot \frac{1}{N_\ell(\mu(t))}  \mathrm{d}  \Lambda^{(\ell,k)}(t)  \nonumber \\
	&~~~~~~~~~~~~~~~~-  \sum_{k=1}^{\ell-1} \int_0^\cdot \frac{1}{N_\ell(\mu(t))}  \mathrm{d}  \Lambda^{(k,\ell)}(t)\,, \qquad \ell = 1, \cdots, d\,.
	    \label{eq:160128.1} 
\end{align}
 Thus, we obtain for the function $\,G = \bm G \circ \mathfrak R\,$ the representations of 
\eqref{eq: thetaG}--\eqref{eq: decor}, with
\begin{align}
	D_i G(x) &= \sum_{\ell=1}^d  \frac{1}{N_\ell(x)}  D_\ell \bm G \big( \mathfrak R(x) \big) \1_{  \,x_{(\ell)}=x_i     } \,, \qquad x \in \mathcal U\,,\,\, \, i = 1, \cdots, d\,; 
	\label{eq:160117.1}\\
	\Gamma^G(\cdot) &= - \frac{1}{2}  \sum_{k=1}^d \sum_{\ell=1}^d \int_0^{ \cdot}    D^2_{k\ell} \bm G \big(\bm \mu (t)  \big)     \mathrm{d} \big \langle\bm  \mu_k,\bm \mu_\ell  \big \rangle (t)  
 -  \sum_{\ell=1}^{d-1}   \sum_{k=\ell+1}^{d}    \int_0^{ \cdot}   \frac{1}{N_\ell(\mu(t))}   D_\ell \bm G \big( \bm\mu ( t) \big)  \mathrm{d}  \Lambda^{(\ell,k)}(t)  
 \nonumber\\
 	&+ \sum_{\ell=2}^{d}   \sum_{k=1}^{\ell-1}    \int_0^{ \cdot}   \frac{1}{N_\ell(\mu(t))}   D_\ell \bm G \big( \bm\mu ( t) \big)  \mathrm{d}  \Lambda^{(k,\ell)}(t). \label{eq:160117.2}
\end{align} 
In particular, $G$ is indeed regular for $\mu(\cdot)$; this confirms the last statement of Theorem 3.6 in this case. 

\medskip
Let us consider now the special case when the    collision local times of order 3 or higher vanish: 
\begin{equation}
\label{eq: coll_LT_bb}
\Lambda^{(k,\ell)} (\cdot)\equiv 0\,; \qquad 1 \le k < \ell \le d\,, \quad   \ell   \ge k+ 2\,.
\end{equation}
This will happen, of course,  when actual triple collisions never occur. It will  also happen   when triple- or higher-order collisions {\it do} occur but are sufficiently ``weak,'' so as not to lead to the accumulation of collision local time; see \citet{Ichiba:Papathanakos:Banner:2013} and \citet{Ichiba:Karatzas:Shkolnikov:2013} for examples of this situation. 
Under \eqref{eq: coll_LT_bb}, only the term  corresponding to  $k = \ell+1$ appears in the  second   summation on the right-hand side of \eqref{eq:160117.2}, and only the term  corresponding to   $k = \ell-1$ appears in the    third summation. 
 \qed
\end{example}

\begin{example}[Regular, but not Lyapunov]   \label{ex:160129}
	Let us consider the function $\bm G:  \mathbb{W}^d \to [0,1]$ defined by $ \bm G (x ):= x_1$. This $\bm G$ is twice continuously differentiable and concave. In particular, as in Example~\ref{ex: smooth}, $\bm G$ is a Lyapunov function for the process $\bm \mu(\cdot)$ in \eqref{eq: R}. 
	
	However, the function $G = \bm G \circ \mathfrak R$, which has the representation $G(x) = \max_{i=1,\cdots,d}\,  x_i\, $ for all $x \in \bm \Delta^d$, is      regular for $\mu(\cdot)\,,$ but {\it typically not Lyapunov.}  Indeed, in the notation of Example~\ref{ex: smooth2}, we have $D_1 \bm G = 1$, $D_\ell \bm G= 0$ for all $\ell = 2, \cdots, d\,,$ and $D_{k \ell}^2 \bm G = 0$ for all $1 \le k,\ell \le d$. Thus, in the notation of  Example~\ref{ex: smooth2}, we have 
	\begin{align*}
	D_i G(x) &\,= \,  \frac{1}{\,\sum_{j=1}^d  \mathbf{ 1}_{ \, x_{(1)} = x_j   }\,} \,  \1_{    \,x_{(1)} =x_i   }\,, \qquad x \in \bm \Delta^d, \,\,\,\, i = 1, \cdots, d 
\end{align*}
as follows directly   from \eqref{eq:160117.1}, and the expression in \eqref{eq:160117.2} simplifies to  
\begin{align*} 
		\Gamma^G(\cdot) &= 
 -    \sum_{k=2}^d    \int_0^{ \cdot}   \frac{1}{ \sum_{i=1}^d  \mathbf{ 1}_{ \{\mu_{(1)} (t)= \mu_i (t)\} }} \,  \mathrm{d}  \Lambda^{(1,k)}(t).
\end{align*} 
Unless the nondecreasing process $\Lambda^{(1,2)}(\cdot)$ is  identically equal to zero, the process $\Gamma^G(\cdot)$ is   non-increasing.  If we now additionally assume the existence of a deflator, then, by Proposition~\ref{pr: unique}, the process $\Gamma^G(\cdot)$ does not depend on the choice of $D G$; thus, $\Gamma^G(\cdot)$ is determined uniquely by the above expression,  so  $G$ cannot be a Lyapunov function for $\mu (\cdot)$.  Example~\ref{large} below generalizes this setup.   \qed
\end{example}

\section{Functionally generated trading strategies} 
 \label{sec: 5}

To simplify notation, and when it is clear from the context, we shall   write from now on $V^\vartheta(\cdot)$ to denote the value process  $V^\vartheta(\cdot \,; \mu)$ given in \eqref{eq:160102.1} for $X(\cdot) = \mu (\cdot)$. Proposition \ref{pr: ChoN} allows us to interpret $V^\vartheta(\cdot)=V^\vartheta(\cdot \,; \mu) = V^\vartheta(\cdot \,; S) / \Sigma (\cdot)$ as the ``relative value"  of the trading strategy $\vartheta (\cdot) \in      \mathscr T (S)$ with respect to the market portfolio.

\subsection{Additive generation}

For any given function $ G : \mathbf{supp\,} ( \mu ) \rightarrow \R$ which is regular for   the vector process $ \mu (\cdot)$ of market weights as in   Definition \ref{def: reg}, we consider the vector   $ {\bm \vartheta}   ( \cdot) =\transpose{\big( {\bm \vartheta}_1 ( \cdot), \cdots,  {\bm \vartheta}_d ( \cdot) \big)} $ of processes $ {\bm \vartheta}_i ( \cdot) \,:= 
D_i G \big( X ( \cdot) \big)$ in \eqref{eq: thetaG}, as well as  the trading strategy $ {\bm \varphi}   ( \cdot) =\transpose{\big( {\bm \varphi}_1 ( \cdot), \cdots,  {\bm \varphi}_d ( \cdot) \big)} $ with components 
\begin{equation}
\label{eq: phiG}
{\bm \varphi}_i ( \cdot) := {\bm \vartheta}_i (\cdot) - Q^{ {\bm \vartheta} } ( \cdot  ) + {\bm C}  , \qquad  i=1, \cdots, d
\end{equation}
 in the manner of \eqref{eq: phi} and \eqref{eq: cue}, and with the real constant
\begin{equation}
\label{eq: cG}
{\bm C} := G  \big( \mu (0) \big) - \sum_{j=1}^d  \mu_j (0)  D_j G  \big( \mu (0) \big). 
\end{equation}

\begin{defn} [Additive functional generation (AFG)]
  \label{def: FG}
We say that the   trading   strategy $ {\bm \varphi}   ( \cdot) = \big( {\bm \varphi}_1 ( \cdot),$  $  \cdots, {\bm \varphi}_d ( \cdot) \big)' \in \mathscr T (\mu)$ of \eqref{eq: phiG} is {\it additively generated} by the regular function  $ G : \mathbf{supp\,} ( \mu ) \rightarrow \R$. \qed
\end{defn}

\begin{rem}[Non-uniqueness of trading strategies]
	There might be two different trading strategies ${\bm \varphi}(\cdot)  \neq \widetilde{\bm \varphi}(\cdot)$,   both generated additively by the same regular function $G$. This is because   the function $DG$ 	
	in Definition~\ref{def: reg} need not be unique. However, if there exists a deflator for $\mu(\cdot)$, then the process $\Gamma^G(\cdot)$ is uniquely determined by 	Proposition \ref{pr: unique}, and \eqref{eq: valphiG} below yields 
	$V^{{\bm \varphi}} (\cdot) = V^{\widetilde {\bm \varphi}} (\cdot)$. \qed
\end{rem}

\begin{prop} [Representation and   value  of  AFG strategies]
 \label{pr: 1}
The  trading  strategy $ {\bm \varphi}   ( \cdot) = \big( {\bm \varphi}_1 ( \cdot), \cdots,  $  ${\bm \varphi}_d ( \cdot) \big)' $, generated additively as in  \eqref{eq: phiG}  by a   regular function $ G : \mathbf{supp\,} (\mu)  \rightarrow \R$, has   relative  value process    
\begin{equation}
\label{eq: valphiG}
V^{{\bm \varphi}} (\cdot ) = G \big( \mu (\cdot) \big) + \Gamma^G (\cdot),
\end{equation}
and can be represented   in the form 
\begin{equation}
\label{eq: phiG2}
{\bm \varphi}_i ( \cdot) \,= \, D_i G \big( \mu (\cdot) \big)+ \Gamma^G (\cdot) +    G \big( \mu (\cdot) \big)   -  \sum_{j=1}^d  \mu_j (\cdot) D_j G  \big( \mu (\cdot) \big)  \,  , \qquad i=1, \cdots, d\,.
\end{equation} 
\end{prop}

\begin{proof}
We substitute from \eqref{eq: phiG} and \eqref{eq: thetaG} into \eqref{eq: valphi}, and recall  \eqref{eq: deco} and \eqref{eq: cG}, to  obtain
$$
V^{{\bm \varphi}} ( \cdot ) =  \sum_{j=1}^d \, \mu_j (0) \, D_j G  \big( \mu (0) \big)  + {\bm C} +  \int_0^{ \cdot} \big \langle \bm \vartheta (t), \mathrm{d} \mu (t) \big \rangle = G \big( \mu (\cdot) \big) + \Gamma^G (\cdot),
$$
that is, \eqref{eq: valphiG}. Using \eqref{eq: phiG}, \eqref{eq: cue},  and \eqref{eq: valphi} we also obtain
\begin{align*}
{\bm \varphi}_i ( \cdot) &= D_i G(\mu(\cdot)) - V^{ {\bm \vartheta} } (\cdot) + V^{ {\bm \vartheta} } (0) +\int_0^{ \cdot} \big \langle \bm \vartheta (t), \mathrm{d} \mu (t) \big \rangle +  \bm C
\\ 
&= D_i G(\mu(\cdot))  - \sum_{j=1}^d  \mu_j (\cdot) D_j G  \big( \mu (\cdot) \big)  + V^{{\bm \varphi}} ( \cdot ), \qquad  i=1, \cdots, d
\end{align*}
leading to \eqref{eq: phiG2}.
\end{proof}

The expression for $\bm \varphi(\cdot)$ in \eqref{eq: phiG2} motivates the interpretation of $\bm \varphi(\cdot)$ as ``delta hedge'' for a given ``generating function'' $G$. Indeed, if we interpret $DG$ as the gradient of $G$, then  for each $i = 1, \cdots, d\,$ and $t \ge 0$ the quantity  $\bm \varphi_i( t)$ is exactly the ``derivative'' $D_i G \big( \mu ( t) \big)$ in the $i$-th direction, plus the global correction term
$$w(t) := V^{\bm \varphi }(t) - \sum_{j=1}^d  \mu_i(t) D_j G(\mu(t)) = \Gamma^G ( t) + G \big( \mu ( t) \big) - \sum_{j=1}^d  \mu_i(t) D_j G(\mu(t))\,, $$
  the same for all $i$, which  ensures the self-financibility of the trading strategy $\bm \varphi(\cdot)$. 

To implement the trading strategy $\bm \varphi(\cdot)$ in \eqref{eq: phiG2} at some time $t \geq 0$, assume 
it has been implemented
up to the present time $t$. It now suffices to compute $ D_i G(\mu(t))$ for each $\,i = 1, \cdots, d\,,$  and to buy exactly $ D_i G(\mu(t))$ shares of the $i$-th asset. If not all wealth gets   invested, that is, if the quantity $w (t)$ is positive, then  
one buys exactly $w(t)$ shares of each asset, costing exactly $\sum_{i=1}^d w(t) \mu_i(t) = w(t)$.  If $w(t)$ is negative,   
one sells  those $|w(t)|$  shares  instead of buying them.  Thus, an {\it implementation of the functionally generated  strategy does not require the computations of any stochastic integral.}

  If the  function $G$ is nonnegative and concave, the following result guarantees that the strategy it generates holds a nonnegative amount of each asset, even if $ D_i G(\mu(t))$ is negative for some $i = 1, \cdots, d$. 

\begin{prop} [Long-only trading strategies]
 \label{pr: 1b}
 Assume that  one of the three conditions in Theorem~\ref{thm: 1} holds for some continuous function $ G :  \mathbf{supp\,} (\mu ) \rightarrow [0,\infty)$.    Then there exists a trading strategy $ {\bm \varphi}(\cdot),$
 additively generated by $G,$ which satisfies
 $ {\bm \varphi}_i ( \cdot) \ge 0$   for each $i = 1, \cdots, d\,;$  in other words, the trading  strategy $ {\bm \varphi}   ( \cdot) $  is then ``long-only.''\end{prop}
 The proof of Proposition~\ref{pr: 1b} requires some convex analysis and is contained in Subsection~\ref{sec: App1} below.

\begin{rem}[Associated portfolios]
\label{cor: 1}
Let $G$ be a  regular function for the vector process $\mu(\cdot)$, generating the  trading  strategy ${\bm \varphi} (\cdot)$ as in \eqref{eq: phiG} and \eqref{eq: phiG2}.
  Whenever $V^{{\bm \varphi}} ( \cdot) >0 $ holds (for example, if $G$ is a Lyapunov function taking values in $(0,\infty)$),  the portfolio weights
\begin{equation}
\label{eq: pi}
{\bm \pi}_i ( \cdot) := \frac{\,  \mu_i ( \cdot)  {\bm \varphi}_i ( \cdot) \,}{V^{{\bm \varphi}} (  \cdot )}  
= \frac{ \mu_i (\cdot) {\bm \varphi}_i ( \cdot) }{ \, \sum_{j=1}^d \mu_j ( \cdot) {\bm \varphi}_j ( \cdot)\,}, \qquad i=1, \cdots, d
\end{equation}
of the trading strategy ${\bm \varphi} (\cdot)$ can be cast as
 \begin{equation}
\label{eq: pi3}
{\bm \pi}_i ( \cdot) =\mu_i ( \cdot) \bigg(  1+\frac{1
}{ G (\mu ( \cdot)) + \Gamma^G ( \cdot) } \bigg( D_i  
G \big( \mu ( \cdot) \big)  - \sum_{j=1}^d  \mu_j ( \cdot) D_j 
 G \big( \mu (  \cdot) \big)  \bigg) \bigg) 
\end{equation}
for each $i=1, \cdots, d$.  \qed
\end{rem}

\subsection{Multiplicative generation}
Let us recall   the  functionally-generated portfolio  introduced by  \citet{F_generating, Fernholz:2001, Fe}. Suppose that the function  $ G : \mathbf{supp\,} (\mu )  \rightarrow [0, \infty) $ is regular for the vector process $ \mu (\cdot)$ of market weights    in \eqref{eq: mw}, and   that $1/G(\mu(\cdot))$ is locally bounded. 
This holds if $G$ is bounded away from zero, or if \eqref{eq: interiors} is satisfied and $G$ is   strictly positive on $\bm\Delta^d_+$. We introduce now the predictable portfolio-weights  
\begin{equation}
\label{eq: Pi}
{\bm \Pi}_i (\cdot) :=\mu_i ( \cdot) 
\bigg(  1 + \frac{1   
}{ G \big( \mu ( \cdot) \big) } \bigg( D_i 
G \big( \mu ( \cdot) \big)  - \sum_{j=1}^d D_j   
G \big( \mu (  \cdot) \big) \, \mu_j ( \cdot)  \bigg)   \bigg), \qquad i=1, \cdots, d.
\end{equation} 
  These processes satisfy $ \sum_{i=1}^d {\bm \Pi}_i (  \cdot)\equiv 1$ rather trivially;  and it is shown as in Proposition~\ref{pr: 1b}   that    
 they  are non-negative, if one of the three conditions in Theorem~\ref{thm: 1} holds.  
 
 In order to relate these portfolio weights to a trading strategy, let us consider the vector process  $ \widetilde{\bm \vartheta}   ( \cdot) =\transpose{\big( \widetilde{\bm \vartheta}_1 ( \cdot), \cdots,  \widetilde{\bm \vartheta}_d ( \cdot) \big)} $ given  in the notation of \eqref{eq: thetaG} by
 \begin{align*}
 	\widetilde{\bm \vartheta}_i(\cdot) := {\bm \vartheta}_i(\cdot) \times  \exp \left(    \int_0^{\cdot} \frac{\mathrm{d} \Gamma^G   (t)    }{  G  \big( \mu (t)\big)   }\right) = D_i G(\mu(\cdot)) \times \exp \left(    \int_0^{\cdot} \frac{\mathrm{d} \Gamma^G   (t)    }{  G  \big( \mu (t)\big)   }\right), \qquad i=1, \cdots, d \,.
 \end{align*}
 Note that the integral  is well-defined, as $1/G(\mu(\cdot))$ is  locally bounded by assumption.  Moreover, we have $ \widetilde{\bm \vartheta} (\cdot) \in \mathscr L (\mu)$ since $ {\bm \vartheta} (\cdot) \in \mathscr L (\mu)$ and the exponential function is locally bounded.  We can turn the predictable process $ \widetilde{\bm \vartheta}   ( \cdot)$ into a trading 
 strategy $ {\bm \psi}   ( \cdot) =\transpose{\big( {\bm \psi}_1 ( \cdot), \cdots,  {\bm \psi}_d ( \cdot) \big)} $   by setting  
\begin{align}  \label{eq: psiG}
{\bm \psi}_i ( \cdot) : =  \widetilde{\bm \vartheta}_i (\cdot) - Q^{ \widetilde {\bm \vartheta} } ( \cdot ) + {\bm C}  , \qquad  i=1, \cdots, d
\end{align}
 in the manner of \eqref{eq: phi} and \eqref{eq: cue}, and with $\bm C$ given by \eqref{eq: cG}.

\begin{defn} [Multiplicative functional generation (MFG)]
  \label{def: MFG}
We say that the  trading   
strategy $ {\bm \psi}   ( \cdot) = \big( {\bm \psi}_1 ( \cdot), \cdots, $  ${\bm \psi}_d ( \cdot) \big)' \in \mathscr T (\mu)$ of \eqref{eq: psiG} is {\it multiplicatively generated} by the  function  $ G : \mathbf{supp\,} (\mu)  \rightarrow [0, \infty)$. \qed
\end{defn}

Proposition~\ref{pr: 1} has now an equivalent formulation.
\begin{prop} [Representation and   value  of   MFG strategies]
 \label{pr: 1m}
The  trading  strategy $ {\bm \psi}   ( \cdot) = \big( {\bm \psi}_1 ( \cdot), \cdots, $  $ {\bm \psi}_d ( \cdot) \big)' $, generated as in  \eqref{eq: psiG}  by a  function  $ G : \mathbf{supp\,} (\mu )  \rightarrow [0, \infty) $ which is regular for the   process $ \mu (\cdot)$ of market weights     and  such that $1/G(\mu(\cdot))$ is locally bounded, has   relative  value process    
 \begin{equation}
\label{FG}
 V^{ {\bm \psi}}  (\cdot   )  \,= \,G  \big( \mu (\cdot)  \big)     \exp \left(    \int_0^{\cdot} \frac{\mathrm{d} \Gamma^G   (t)    }{  G  \big( \mu (t)\big)   }\right)\,>\,0
  \end{equation}
and can be represented   in the form 
\begin{equation}
\label{eq: Pi2}
 {\bm \psi}_i ( \cdot)=  
V^{{\bm \psi}} (  \cdot  )    \bigg(  1 + \frac{1   
}{ G \big( \mu ( \cdot) \big) } \bigg( D_i 
G \big( \mu ( \cdot) \big)  - \sum_{j=1}^d D_j   
G \big( \mu (  \cdot) \big) \mu_j ( \cdot)  \bigg)   \bigg), \qquad i=1, \cdots, d.
 \end{equation} 
\end{prop}
\begin{proof}
	With  $K(\cdot) :=  \exp \left(    \int_0^{\cdot} 
	\big( 1/G(\mu( t))\big)\, \mathrm{d} \Gamma^G   (t)   
	\right)$,  
	the product rule yields
\begin{align*}
		\mathrm{d} \big(  G   ( \mu (t)   )   K(t)   \big)
			 &=  K(t)  \mathrm{d}  G  \big( \mu (t)  \big) + K(t) \mathrm{d}  \Gamma^G  (t)  
			= K(t) \sum_{i=1}^d  \bm \vartheta_i(t) \mathrm{d}  \mu_i(t) 		\\
&= \sum_{i=1}^d  \widetilde{\bm \vartheta}_i(t) \mathrm{d}  \mu_i(t)		=\sum_{i=1}^d  \bm \psi_i(t) \mathrm{d}  \mu_i(t) = \mathrm{d}  V^{ {\bm \psi}}  (t    ), \qquad t \geq 0,
\end{align*}
	where the second equality uses \eqref{eq: deco}, and the second-to-last 
	relies on \eqref{eq: valphi}.  Since  \eqref{FG} holds at time zero, namely
	$\,
 V^{ {\bm \psi}}  (0  )	= \sum_{i=1}^d {\bm \psi}_i ( 0)\mu_i (0) = \sum_{i=1}^d \big(  {\bm \vartheta}_i ( 0) + \bm  C \big) \mu_i (0)= G \big( \mu (0) \big)
	\,$
on the strength of \eqref{eq:160102.1}, \eqref{eq: psiG},   \eqref{eq: cue}  and \eqref{eq: cG}, 	it follows from the above display that \eqref{FG} holds     in general.  
	
	On the other hand, starting with \eqref{eq: psiG}  we obtain
	\begin{align*}
{\bm \psi}_i ( \cdot) &=  \,\widetilde{\bm \vartheta}_i (\cdot) - Q^{ \widetilde {\bm \vartheta} } ( \cdot ) + {\bm C} 
= K(\cdot) D_i G(\mu(\cdot)) - V^{\widetilde{\bm \vartheta}} (\cdot) + V^{\widetilde{\bm \vartheta}} (0) +  \int_0^{\, \cdot}  \big \langle \widetilde{\bm \vartheta} (t), \mathrm{d} \mu (t) \big \rangle + \bm C\\
&= \,K(\cdot) D_i G(\mu(\cdot)) - K(\cdot) \sum_{j=1}^d D_j   
G \big( \mu (  \cdot) \big) \mu_j ( \cdot)   +  V^{ {\bm \psi}}  (\cdot    ), \qquad \qquad  i=1, \cdots, d\,,
\end{align*}
using \eqref{eq: valphi} and the definition of $Q^{ \widetilde {\bm \vartheta} } ( \cdot )$ in \eqref{eq: cue}. This yields the representation \eqref{eq: Pi2}.
\end{proof}
It is now easy to see how the portfolio process ${\bm \Pi}(\cdot)$ in \eqref{eq: Pi} is obtained from \eqref{eq: Pi2} in the same manner as \eqref{eq: pi}, as $V^{\bm \psi}(\cdot)$ is strictly positive.
The  representation in \eqref{FG} 
 is a ``generalized  master equation" in the spirit of Theorem~3.1.5 in \citet{Fe}. 
  
  \subsection{Comparison of additive and multiplicative functional generation}
  
  It is instructive at this point to compare additive and multiplicative functional generation.    On a purely formal level, the multiplicative generation of Definition~\ref{def: MFG} requires a regular function $G$ such that $1/G(\mu(\cdot))$ is locally bounded.  On the other side,   additive functional generation  requires only the regularity of the function $G$.  
  
  At time $t=0,$  the additively-generated  strategy agrees with the multiplicatively-generated one; that is, we have $\bm \varphi(0) = \bm \psi(0)$ in the notation of \eqref{eq: phiG2} and \eqref{eq: Pi2}. However, at any   time $t>0$   with $\Gamma^G(t) \neq 0$, these two strategies usually differ; this is seen  most easily   by looking at their corresponding portfolios \eqref{eq: pi3} and \eqref{eq: Pi}.  
  More precisely, the two strategies differ in the way they allocate the proportion of their wealth captured by $\Gamma^G(\cdot)$.  The additively-generated strategy tries to allocate this proportion uniformly across all assets in the market; whereas the multiplicatively-generated strategy tends to correct for this amount by proportionally adjusting the asset holdings. 
  
  To see this, consider again \eqref{eq: Pi2} and assume that $\sum_{j=1}^d x_j  D_j   
G ( x)   = G(x)$   for all $x \in \bm \Delta^d$, a case which occurs often in examples; we have then
\begin{align*}
 {\bm \psi}_i ( \cdot)\,=  \,D_i G \big( \mu ( \cdot) \big) \,
 \exp \left(    \int_0^{\cdot} \frac{\mathrm{d} \Gamma^G   (t)    }{  G  \big( \mu (t)\big)   }\right)
, \qquad i=1, \cdots, d.
 \end{align*} 
 Thus,  in this situation, the multiplicatively-generated $\bm{\psi}(t)$ does not invest in assets for which $D_i G(\mu(t)) = 0$, for each $t \geq 0$,
 but instead adjusts the holdings proportionally.   By contrast, the additively-generated $\bm{\varphi}(\cdot)$ buys shares 
 \begin{align*}
 {\bm \varphi}_i (  t)\,=  \,D_i G \big( \mu (  t) \big) + \Gamma^G   (t)
, \qquad i=1, \cdots, d
 \end{align*}
 of the different assets in this case, and does {\it not} shun stocks for which $D_i G(\mu(t)) = 0$, at time $t$.

 \medskip
 \noindent
 {\it Ramifications:}   This difference in the two strategies leads to two observations.  
  
  First, if one is interested in a trading strategy that invests through  time only in a subset of the market, such as for example the set of ``small-capitalization stocks", then   strategies generated multiplicatively  by   functions $G$ that satisfy  $\sum_{j=1}^d x_j  D_j   
G ( x)    = G(x)$  for all $x \in \bm \Delta^d$, are appropriate. If, on the other hand, one wants to invest the trading strategy's earnings in a proportion of the whole market, additive generation is better suited. This is illustrated by Examples~\ref{large} and \ref{small}.   

Secondly, the trading strategy which  holds   equal weights across all assets, can be  generated multiplicatively, by the ``geometric mean" function $\bm \Delta^d \ni x \mapsto G(x) = (x_1 \times \cdots \times x_d)^{1/d} \in (0,1)\,$ 
as long as \eqref{eq: interiors} holds; indeed, the portfolio weights in (\ref{eq: Pi}) become now $\, {\bm \Pi}_i (\cdot) = 1/d\,$ for all $\, i=1, \cdots, d\,.$ 
 But such a trading strategy cannot be  additively  generated; for instance, the portfolio in \eqref{eq: pi3}, namely 
 $$
 {\bm \pi}_i (t) \,=\, \frac{\, ( 1 / d) + \mu_i (t) R^G (t)\, }{1 + R^G (t)}\,, \qquad i=1, \cdots, d\,, \quad t \ge 0\qquad \text{with}\qquad R^G (t) := \frac{\, \Gamma^G (t)\,}{\, G (\mu (t))\,}\,,
 $$
 that corresponds to the  strategy    
 generated additively by this geometric-mean function $G,$   distributes the gains described by $\Gamma^G(\cdot)$ uniformly across stocks, 
and  this destroys   equal weighting.
  
    \smallskip

   \medskip
 \noindent
 {\it Comparison of portfolios:}
Let us compare  the  two portfolios in \eqref{eq: pi3} and \eqref{eq: Pi} more closely. These differ only in the denominators that appear inside the brackets on their right-hand sides. Computing the quantities of \eqref{eq: Pi}  needs, at any given time $t \geq 0$,  knowledge of   the   configuration of market weights $ \mu_1 (t), \cdots, \mu_d (t)$ prevalent at that time -- {\it and   nothing else}. By contrast, the quantities of \eqref{eq: pi3} need the entire history of these market weights during the interval $[0,t]$, in order to compute the integral in \eqref{eq: deco}. When these portfolios are expressed as trading strategies, as is done in \eqref{eq: phiG2} and \eqref{eq: Pi2}, then in both cases only the wealth process  and the market weights $ \mu_1 (t), \cdots, \mu_d (t)$ are necessary.

\section{Sufficient conditions for  outperformance} 
 \label{sec: 6}

We have developed by now the necessary machinery in order to present sufficient conditions for the possibility of outperforming the market, as introduced in Subsection~\ref{sec: 3b} -- at least over sufficiently long time horizons.

In this section, $G: \mathbf{supp\,} (\mu)  \rightarrow [0,\infty)$ denotes a nonnegative 
regular function for the market-weight process $ \mu (\cdot)$ with $G(\mu(0)) = 1$. This  normalization ensures that the initial wealth of a functionally generated strategy starts with one dollar, as required by \eqref{eq:160130.1}; see \eqref{eq: valphiG} and \eqref{FG}.  Such a normalization can always be achieved upon replacing $\,G$ by $\,G+1\,$ if $G(\mu(0)) = 0,$ or by  
$\,G/G(\mu(0))\,$ if $G(\mu(0)) > 0 $.

\begin{thm}[Additively  generated outperformance]
 \label{thm: 2}
 Fix a Lyapunov  function $G: \mathbf{supp\,} (\mu)  \rightarrow [0,\infty)$    satisfying   $G(\mu(0)) = 1,$ and suppose that  for some  real number  $ T_* >0$    we have
\begin{equation}   
 \label{eq:160131.3}
\P \big( \Gamma^G (T_*) > 1 \big) = 1 .
\end{equation}
 Then the additively  generated   strategy $ {\bm \varphi}  ( \cdot) = \transpose{(   {\bm \varphi}_1  ( \cdot), \cdots,   {\bm \varphi}_d  ( \cdot) )} $ 
 of Definition~\ref{def: FG}
   strongly outperforms the market   over every time-horizon $[0,T] $ with $T \geq T_*$.   
 \end{thm} 
\begin{proof} 
We recall the observations in Remark~\ref{rem: 1} and note that  \eqref{eq: valphiG} yields $V^{  {\bm \varphi}} (0 ) =1$, $V^{  {\bm \varphi}} (\cdot ) \geq 0$, and $
V^{  {\bm \varphi}} ( T    ) = G \big( \mu ( T ) \big) + \Gamma^{  G} ( T  )\ge \Gamma^G (T_*) > 1$ for all $T \geq T_*$.
  \end{proof}
 
  The following result  complements Theorem~\ref{thm: 2}.
  
  \begin{thm}[Multiplicatively  generated outperformance]
 \label{thm: 3}
Fix a regular  function $G: \mathbf{supp\,} (\mu)  \rightarrow [0,\infty)$    satisfying   $G(\mu(0)) = 1,$ and suppose that  for some  real numbers $ T_* >0$  and  $\varepsilon>0$ we have
\begin{equation*}
\P \big( \Gamma^G (T_*) > 1 + \varepsilon \big) = 1\,.
\end{equation*}
Then there exists a constant $c>0$ such that the trading strategy
$ {\bm \psi}^{(c)}  ( \cdot) = \transpose{(   {\bm \psi}^{(c)}_1  ( \cdot), \cdots,   {\bm \psi}^{(c)}_d  ( \cdot) )}$,
multiplicatively generated by the  regular function $G^{(c)} := (G+c) / (1+c)$ as in Definition~\ref{def: MFG},
     strongly outperforms the market   over the time-horizon $[0,T_*] ;$  and, if $G$ is a Lyapunov function, also    over every time-horizon $[0,T] $ with $T \geq T_*$.
 \end{thm} 
\begin{proof} 
For $c>0$, the representation  \eqref{FG} yields the comparisons  $V^{  {\bm \psi}^{(c)}} (0 ) =1$, $V^{  {\bm \psi}^{(c)}} (\cdot ) > 0$, and
\begin{equation} 
\label{eq:160131.1}
V^{  {\bm \psi}^{(c)}} ( T_*    ) \geq \frac{c}{1+c} \times \exp \left(    \int_0^{T_*} \frac{\mathrm{d} \Gamma^G   (t)    }{  G  \big( \mu (t)\big)   + c}\right) > \frac{c}{1+c} \times \exp \left(     \frac{1+\varepsilon   }{ \bm \kappa  + c}\right),
\end{equation}
where $\bm \kappa$ is an upper bound on  $G$, which is assumed to be continuous on the compact set $\mathbf{supp\,} (\mu)$. Here, we used  in the first inequality the bound  $G \geq 0$ and the identity $\Gamma^{G^{(c)}}(\cdot) = \Gamma^G(\cdot)/(1+c)$.
Now, with the help of Remark~\ref{rem: 1} we may conclude again, as soon as we have argued the existence of a constant $c>0$ such that the last term in \eqref{eq:160131.1} is greater than one. Taking logarithms yields
\begin{align} \label{eq:160131.2}
	-\log\left(1 + \frac{1}{c}\right) + \frac{1+\varepsilon}{\bm \kappa +c} >
		\frac{\varepsilon - \bm \kappa \log\left(1 + 1/c\right) }{\bm \kappa+c}
\end{align}
for all $c >0,$ since  
$1>c \log(1+1/c)$. However, the right-hand side of \eqref{eq:160131.2} is positive for sufficiently large $c$, and this concludes the proof.

If $G$ is a Lyapunov function, $\,\P \big( \Gamma^G (T ) > 1 + \varepsilon \big) = 1\,$ and the inequalities in \eqref{eq:160131.1} are valid for all   $\,T \ge T_*\,,$   and the same reasoning as above works once again. 
  \end{proof}

We illustrate now the  previous two theorems with two examples.

  \begin{example} [Entropy function and excess growth]
    \label{entropy}
Consider the {Gibbs} {\it entropy function }
$$
H (x) = \sum_{i=1}^d  x_i \log \left( \frac{1}{x_i}\right), \qquad x \in {\bm \Delta}^d
$$
 with values in $ [0, \log (d)]$  and the understanding $ 0 \times \log (\infty)  = 0$.  
 This $H$ is concave and continuous on $\bm \Delta^d$ and strictly positive on ${\bm \Delta}^n_+\,.$   It is a Lyapunov function for $\mu(\cdot)$ provided that, as we  assume from now on in this example,   either   a deflator for $\mu(\cdot)$ exists, or \eqref{eq: interiors} holds; cf. Theorem~\ref{thm: 1}\ref{thm: 1i}\&\ref{thm: 1iii}.
 
 Elementary computations then show  that the process of \eqref{eq: Gamma} takes now the form 
 \begin{equation*}
\Gamma^H (\cdot) = 
\frac{1}{2}   \sum_{i=1}^d \int_0^{ \cdot} \mathbf{ 1}_{ \{ \mu_i (t) >0\}   }\frac{    \mathrm{d}    \big \langle  \mu_i   \big  \rangle  (t)}{\mu_i (t)}  =
\frac{1}{2}   \sum_{i=1}^d \int_0^{ \cdot}  \mu_i (t)    \mathrm{d}    \big \langle \log (\mu_i)    \big  \rangle (t) .
\end{equation*}
This is the    {\it cumulative excess growth of the market,} a trace-like quantity which  plays a very important role in Stochastic Portfolio Theory. It  measures the market's cumulative  ``relative variation'' -- stock-by-stock, then  averaged  according  to each stock's market weight. It is easy to see that $\Gamma^H(\cdot)$ is clearly non-decreasing, which confirms that the {Gibbs} entropy is indeed a {Lyapunov} function for any market $\mu(\cdot)$ that allows for a deflator or satisfies \eqref{eq: interiors}.

\smallskip
The additively-generated strategy $\bm \varphi(\cdot)$ of \eqref{eq: phiG2}  invests   a number 
$$
{\bm \varphi}
_i ( \cdot)=\left(  \log \left( \frac{1} {\mu_i ( \cdot) }\right) + \Gamma^H ( \cdot)  \right) \1_{\{  \mu_i(\cdot)>0\}}\,,   \qquad i=1, \cdots, d
$$
of shares in each of the various assets, and generates strictly positive value $$V^{{\bm \varphi}} ( \cdot) =H (\mu (\cdot)) + \Gamma^H (\cdot)>0 \,.$$ This strict positivity is obvious if \eqref{eq: interiors} holds; on  the other hand, to see this assuming the existence of a deflator,    consider the stopping time $\tau := \inf \{ t \ge 0 : V^{{\bm \varphi}} (  t) =0\}>0$ on the strength of $\mu (0) \in {\bm \Delta}^n_+\,.$ On the event $\{ \tau < \infty\}$ we have both $H (\mu (\tau) ) =0$ and $\Gamma^H (\tau) =0$. From the properties of the entropy function, the first of these requirements implies that, at time $\tau$, the process of market weights is at one of the vertices of the simplex: $\tau \ge \mathscr{D}_*$ in the notation of (\ref{eq: 6*}). The second requirement gives   $\Gamma^H (\mathscr{D}_*) =0$, thus    $\Gamma^H (\mathscr{D}) =0$  in the notation of (\ref{eq: 6}). But then $$2 \, \Gamma^H (\mathscr{D}) \,=\, \sum_{i=1}^d\int_0^{\mathscr{D}} \frac{d \langle \mu_i \rangle (t)}{\mu_i (t)} \,\ge\, \sum_{i=1}^d  \, \langle \mu_i \rangle (\mathscr{D})  $$ implies that, for each $i=1, \cdots, d\,,$   we have $\langle \mu_i \rangle (\mathscr{D})=0$    on the event $\{ \tau < \infty\}$; the existence of a deflator   leads to $\mu_i   (t)=\mu_i   (0)$ for all    $0 \le t \le \mathscr{D}$, and this to $\mathbb{P} (\tau < \infty)=0$.

\smallskip
Multiplicative generation needs a regular function that is bounded away from zero, so let us consider $H^{(c)} = H + c$ for some $c>0$.  According to \eqref{eq: Pi2},  the multiplicatively generated strategy invests a number
$$
{\bm \psi}^{(c)}_i ( \cdot)=  \left(\log \left( \frac{1} {\mu_i ( \cdot) }\right)  + c \right)\times \exp \left(    \int_0^{\cdot} \frac{\mathrm{d} \Gamma^H   (t)    }{  H  \big( \mu (t)\big)   + c}\right) \1_{\{  \mu_i(t)>0\}}\,, \qquad i=1, \cdots, d
$$
of shares in each of the various assets.

We can  compute now the  portfolio weights corresponding to these two strategies from  \eqref{eq: pi3}  and \eqref{eq: Pi}, respectively, as
 \begin{align*}
{\bm \pi}_i ( \cdot) &= \frac{\mu_i (\cdot) }{H (\mu (\cdot)) + \Gamma^H (\cdot)}  \left( \log \left( \frac{1}{ \mu_i (\cdot)} \right) + \Gamma^H (\cdot)  \right),
\qquad i=1, \cdots, d\,,   \\
{\bm \Pi}_i^{(c)} (\cdot) &=   \frac{\mu_i (\cdot) }{H (\mu (\cdot))+c}  \left(  \log \left( \frac{1}{ \mu_i (\cdot)} \right) + c\right), \qquad i=1, \cdots, d\,,
\end{align*}
with the previous understanding $0 \times \log (\infty)  = 0$. The process ${\bm \Pi}^{(c)}(\cdot)$ has been termed ``entropy-weighted portfolio''  in the literature; see \cite{Fe}, and \citet{FKVolStabMarkets}. 

\smallskip
Let us now consider the question of outperformance. By definition, a trading strategy that  strongly outperforms the market starts with   wealth of one dollar; cf.$\,$\eqref{eq:160130.1}. Hence we shall consider the   Lyapunov function $G = H/H(\mu(0))$ along with its nondecreasing process 
 $\Gamma^G(\cdot)  = \Gamma^H(\cdot)/H(\mu(0))$.  Then Theorems~\ref{thm: 2} and \ref{thm: 3} yield the existence of a such a strategy over the time horizon $[0,T]$, as long as we have, respectively, 
  \begin{equation*}
\P \big( \Gamma^H (T) > H(\mu(0)) \big) = 1\,,
\end{equation*}
or
\begin{equation*}
\P \big( \Gamma^H (T) > H(\mu(0)) + \varepsilon \big) = 1
\end{equation*}
for some $\varepsilon>0$. In the first case, this strong outperformance is additively generated through the trading strategy $\bm \varphi(\cdot)/H(\mu(0))$; in the second, it is multiplicatively generated through the trading strategy $\bm \psi^{(c)}(\cdot)/(H(\mu(0))+c)$ for some sufficiently large $c > 0$. 

\smallskip
For example, if  $ \P  \big( \Gamma^H (t) \ge \eta \,  t , \,\, \forall \,\,  t \geq 0 \big) =1$ holds for some real constant $ \eta >0,$  strong   
outperformance of the market can be implemented over any time-horizon $[0,T]$ with $\,
T> H (\mu (0))/\eta .$ 
\qed
 \end{example}

 \begin{rem}[An old question]
\label{rem: OQ}
  It has been a long-standing open problem, dating to 
\citet{FKVolStabMarkets}, whether the validity of $\P  \big( \Gamma^G (t) \ge \eta \, t , \,\, \forall \,\, t \geq 0 \big) =1$   for some real constant $ \eta >0,$ can guarantee the existence of a strategy that outperforms the market  over {\it any} time-horizon $[0,T],$ of {\it arbitrary} length $T\in (0, \infty)$. 
For explicit examples   showing that this is not possible   in general, see our companion paper \citet{Fernholz:Karatzas:Ruf:2016}.
\qed
\end{rem}
 
 \begin{example} [Quadratic   function and sum  of variations] 
\label{Quad}
Fix, for the moment,  a constant $c \in \R$  and consider, in the manner of Example 3.3.3 in \cite{Fe}, the    quadratic function  $$H^{(c)}(x) \,:= \,c -   \sum_{i=1}^d  x_i^2\,, \qquad  x \in {\bm \Delta}^d, $$  with    values in $ [ c-1, c- 1/d]$.   The term $\sum_{i=1}^d \mu_i^2(\cdot)$ is the weighted average capitalization of the market and may be used to quantify the concentration of capital in a market.

Clearly, $H^{(c)}$ is concave and Theorem~\ref{thm: 1}\ref{thm: 1ii}, or alternatively, Example~\ref{ex: smooth}, yields that $H^{(c)}$ is a Lyapunov function for $\mu(\cdot)$, without any additional assumption. 
The nondecreasing process   of \eqref{eq: deco} is then given by 
$$
		\Gamma^{H^{(c)}}( \cdot)= 
		\sum_{i=1}^d   \big \langle \mu_i 
		\big \rangle( \cdot)$$   
and the additively generated strategy $\bm \varphi^{(c)}(\cdot)$ of \eqref{eq: phiG2}   is given by 
\begin{align*}
{\bm \varphi}^{(c)}
_i ( \cdot)= c - 2   \mu_i ( \cdot)   +  \sum_{j=1}^d \Big( \langle \mu_j \rangle ( \cdot)+ \big( \mu_j  ( \cdot) \big)^2 \Big) 	, \qquad i=1, \cdots, d,.
\end{align*}
If  $c>1$, the multiplicatively generated strategy $\bm\psi^{(c)}(\cdot)$ of \eqref{eq: Pi2} is well-defined and is given as 
\begin{align*}
	\bm \psi^{(c)}_i( \cdot) &= K^{(c)} (\cdot) \left(-\,2 \mu_i  (\cdot)+  \sum_{j=1}^d \big(\mu_j (\cdot) \big)^2 + c   \right)	, \qquad i=1, \cdots, d,
\end{align*}
where 
\begin{align*}
	K^{(c)}(\cdot) = \exp \left(    \int_0^{\cdot} \frac{\mathrm{d} \Gamma^{H^{(c)}}   (t)    }{  H^{(c)}  \big( \mu (t)\big)   }\right) =  \exp\left(\sum_{i=1}^d \int_0^\cdot \frac{\mathrm d \langle\mu_i\rangle(t)}{c - \sum_{j=1}^d \big( \mu_j (t) \big)^2}\right).
\end{align*}

Since $H^{(1)} \geq 0$, we obtain  as in Example~\ref{entropy}  that the condition  
  \begin{equation} \label{ex2:eq1}
\P \left( \sum_{i=1}^d \langle \mu_i\rangle (T) > H^{(1)}(\mu(0)) \right) = 1
\end{equation}
yields a strategy which strongly outperforms the market   on $[0,T]$, and is additively generated by the function $H^{(1)}/H^{(1)}(\mu(0))$. Moreover, the requirement 
\begin{equation} \label{ex2:eq2}
\P \left(  \sum_{i=1}^d \langle \mu_i\rangle (T)> H^{(1)}(\mu(0)) + \varepsilon \right) = 1
\end{equation}
for some $\varepsilon>0\,,$ yields a strategy which strongly outperforms the market  on $[0,T]$, and is multiplicatively generated by the function  $H^{(c)}/H^{(c)}(\mu(0))$ for some sufficiently large $c >1$.  

\smallskip
For example, if
 $ \P \big( \sum_{i=1}^d   \big \langle \mu_i 
		\big \rangle( t ) \ge \eta  t , \,\, \forall \,\,   t \geq 0 \big)=1$ holds,  then     there exist   both additively and multiplicatively generated strong outperformance of the market over {\it any} time-horizon $[0,T]$ with 
\begin{equation} \label{ex2:eq3}
T > \frac{1}{ \eta}  \bigg(    
1- \sum_{i=1}^d \big( \mu_i (0) \big)^2  \,\bigg).
\end{equation}

Let us assume now that the market is diverse, namely, $$\, \max_{\,i = 1, \cdots, d} \, \mu_i( t)  \,< \, 1- \delta \,, \qquad   t \ge 0 $$ holds for some real constant $\delta \in (0,1)$. Then we have the bound $H^{(c)} \geq c-1 + 2\delta(1-\delta)$. Thus, in particular, $H^{(1-2\delta(1-\delta))} \geq 0$ and we may replace $H^{(1)}$ in \eqref{ex2:eq1} and \eqref{ex2:eq2} by $H^{(1-2\delta(1-\delta))}$. This in turn allows us to replace the bound in \eqref{ex2:eq3} by the improved bound
\begin{equation*} 
T > \frac{1}{ \eta}  \bigg(    
1-2\delta(1-\delta)- \sum_{i=1}^d \big( \mu_i (0) \big)^2  \,\bigg).
\end{equation*}

Finally, for future reference, we remark that the   modification 
\begin{equation}
\label{eq: pre-Gini}
H^\flat (x) := 1 - \frac{1}{2} \sum_{i=1}^d \bigg(  x_i - \frac{1}{d}  \bigg)^2, \qquad x \in {\bm \Delta}^d
\end{equation}
of the above quadratic function,  satisfies $H^\flat = H^{(2 + 1/d)}/2$.
\qed
\end{example}

 \section{Further examples}
  \label{sec: 7}
 In this section, we collect several examples that illustrate a variety of Lyapunov functions and their corresponding trading strategies.
 
 \begin{example} [Gini   function and sum of local times] 
\label{Gino}
Let us revisit  Example~4.2.2 of \citet{Fe} in our context. We consider the {\it {Gini} function} 
\[
G^\flat (x) := 1 - \frac{1}{2} \sum_{i=1}^d  \bigg| x_i - \frac{1}{d}  \bigg|, \qquad x \in {\bm \Delta}^d,
\]
which is concave on $ {\bm \Delta}^d$. Thanks to  Theorem~\ref{thm: 1}\ref{thm: 1ii}  $G^\flat$ is 
a Lyapunov function.

 This function is used widely as a measure of inequality; the quadratic function of \eqref{eq: pre-Gini} is its ``smooth sibling.'' For this {Gini} function, and with the help of   the {It\^o-Tanaka} formula, the processes of \eqref{eq: thetaG} and \eqref{eq: deco}  take the form  
$$
 {\bm \vartheta}_i^{G^\flat} ( \cdot) =   - \frac{1}{  2}  \,  \text{sgn} \Big( \mu_i ( \cdot) - \frac{1}{\,d\,} \Big), \quad i=1, \cdots, d \qquad  \text{and} \qquad 
\Gamma^{G^\flat} (\cdot)  = \sum_{i=1}^d  \Lambda_i(\cdot),
$$
respectively. Here $ \Lambda_i (\cdot) $ stands for the local time accumulated by the process $ \mu_i (\cdot)$ at the point $1/d,$ and ``sgn''   for the left-continuous version of the signum function. It is now fairly easy to write down   the strategies of \eqref{eq: phiG2} and \eqref{eq: Pi2} generated by this function. It is  harder, though, to posit  a condition of the type \eqref{eq:160131.3}, as the sum  of local times  $ \sum_{i=1}^d  \Lambda_i (\cdot)$ does not typically admit a strictly positive lower bound.  
\qed
 \end{example}
 
 In the following we present  examples of functional generation of trading strategies based on ranks. 
 
\begin{example} [Capitalization-weighted portfolio of large stocks]
     \label{large}
Let us recall the notation of \eqref{eq: Wn}, fix an integer $ m\in \{ 1, \cdots, d-1\}$  and consider, in the manner of Example 4.3.2 in \cite{Fe}, the function $ \mathbf{ G}^L:\mathbb W^d \rightarrow (0,\infty)$ given by $$\bm G^L\big(x_1, \cdots, x_d \big)\, :=\, x_1 + \cdots + x_m \,.$$ 
If $m=1$ then we are exactly in the setup of Example~\ref{ex:160129}. The function $G^L :=  \mathbf{ G}^L \circ \mathfrak R$ in the notation of \eqref{eq: R} is regular, thanks to Theorem~\ref{thm: 1b} or, alternatively, Example~\ref{ex: smooth2}. In the notation of that example, the corresponding function $DG^L$ can by computed by \eqref{eq:160117.1} as
\begin{align*}
	D_i G^L(x) &\,=\, \sum_{\ell=1}^m  \frac{ 1   }{N_\ell(x)}   
	\1_{ \,x_{(\ell)}=x_i}\, =  \, \1_{\, x_{(m+1)}<x_i} +    \frac{\sum_{\ell=1}^m  \1_{ \, x_{(\ell)}=x_i} }{\sum_{\ell=1}^d  \1_{\, x_{(\ell)}=x_i}}   
\,	 \1_{\, x_{(m+1)}=x_i}
\end{align*}
for all $\, x \in \bm \Delta^d,\,\, \, i = 1, \cdots, d.$ Thanks to \eqref{eq:160117.2}, the process $\Gamma^{G^L}(\cdot) $ is given by
\begin{align*}
	\Gamma^{G^L}(\cdot) &=   \sum_{\ell=2}^{m}   \sum_{k=1}^{\ell-1}    \int_0^{ \cdot}   \frac{1}{N_\ell(\mu(t))}    \mathrm{d}  \Lambda^{(k,\ell)}(t)
 -  \sum_{\ell=1}^{m}   \sum_{k=\ell+1}^{d}    \int_0^{ \cdot}   \frac{1}{N_\ell(\mu(t))}    \mathrm{d}  \Lambda^{(\ell,k)}(t)\\
 	&= \sum_{\ell=1}^{m-1}   \sum_{k=\ell+1}^{m}    \int_0^{ \cdot}   \frac{1}{N_\ell(\mu(t))}    \mathrm{d}  \Lambda^{(\ell,k)}(t)
 -  \sum_{\ell=1}^{m}   \sum_{k=\ell+1}^{d}    \int_0^{ \cdot}   \frac{1}{N_\ell(\mu(t))}    \mathrm{d}  \Lambda^{(\ell,k)}(t)\\
 &= -\sum_{\ell=1}^{m}   \sum_{k=m+1}^{d}    \int_0^{ \cdot}   \frac{1}{N_m(\mu(t))}    \mathrm{d}  \Lambda^{(\ell,k)}(t).  
\end{align*} 
Here the second equality swaps the summation in the first term, relabels the indices, and uses the fact that $N_\ell(\mu(\cdot)) = N_k(\mu(\cdot))$ holds on the  support  of the collision local time $\Lambda^{(\ell,k)}(\cdot)$, for each $1 \le \ell  < k  \le d$.  The last equality used the fact that $N_\ell(\mu(\cdot)) = N_m(\mu(\cdot))$ holds  on the support of  $\Lambda^{(\ell,k)}(\cdot)$, for each $\ell = 1, \cdots, m$ and $k = m+1, \cdots, d$.

If there are no triple points at all, that is, if $ \mu_{(\ell)}(\cdot) - \mu_{(\ell+2)}(\cdot)> 0$   
holds for all $\ell = 1, \cdots, d-2$, then $N_m(\mu(\cdot)) \in \{1,2\}$ and we get   
\begin{align*}
	D_i G^L(x) &=    \1_{\, x_{(m+1)}<x_i } +    \frac{1}{\,2\,}\, \1_{ \, x_{(m)}=x_{(m+1)}=x_i}\,, \qquad x \in \bm \Delta^d, \, \,\,\,i = 1, \cdots, d\,;\\
	\Gamma^{G^L}(\cdot) &= -\frac{1}{2} \, \Lambda^{(m,m+1)}(\cdot).
\end{align*}
For the additively-generated strategy $\bm \varphi(\cdot)$ in \eqref{eq: phiG2} we get
\begin{align*}
	\bm \varphi_i(\cdot) = D_i G^L(\mu(\cdot)) + \Gamma^{G^L}(\cdot), \qquad  \, i = 1, \cdots, d\,;
\end{align*}
and for the multiplicatively-generated strategy $\bm \psi(\cdot)$ in \eqref{eq: Pi2} we have
\begin{align*}
	\bm \psi_i(\cdot) = D_i G^L(\mu(\cdot)) \times \exp\left(\int_0^\cdot \frac{\mathrm{d} \Gamma^{G^L}(t)}{G^L(\mu(t))}\right), \qquad  \, i = 1, \cdots, d.
\end{align*}
Hence, the additively-generated strategy invests in all assets (possibly by selling them), provided that $\Gamma^{G^L}(\cdot)$ is not identically equal to  zero; while the multiplicatively-generated strategy only invests in the $m$ largest stocks.
Whereas the additively-generated strategy might lead to negative wealth, the multiplicatively-generated strategy yields always strictly positive wealth; cf.$\,$\eqref{FG}. Thus, we may express the multiplicatively-generated strategy $\bm \psi(\cdot)$ in terms of proportions, as in \eqref{eq: Pi}, by
$$
 {\bm \Pi}_{
 i}^{  G^L } (\cdot) =\frac{D_i {G^L} (\mu(\cdot))}{\mu_{(1)} (\cdot) + \cdots + \mu_{(m)} (\cdot)}, \qquad i=1, \cdots, d.
$$
We note that this trading strategy only invests in the $m$ largest stocks, and 
 in proportion to each   of these stocks' capitalization, apart from the times when several stocks share the $m$-th position, in which case the corresponding capital is uniformly distributed over these stocks.
 
 In the context of the present example we might think of $d=7,500$ as the entire US market; and of $ m=500,$ as in S\&P 500.   Alternatively, we might consider $m=1,$ when we are adamant about investing only in the market's biggest company.  The non-increasing process $\Gamma^{G^L}(\cdot)$ captures the ``leakage''   that such a trading strategy suffers every time it has to sell -- at a loss -- a stock that has dropped out of the higher-capitalization index and been relegated to the ``minor (capitalization) leagues.'' 
\qed
 \end{example}

  \begin{example}[Capitalization-weighted portfolio of small stocks]
     \label{small}
     Instead of large stocks, as in Example~\ref{large}, we now consider a portfolio consisting of stocks with small capitalization.  With the notation recalled in the previous example, we fix  again an integer $ m\in \{ 1, \cdots, d -1\}$ and consider the function $ \mathbf{ G}^S:\mathbb W^d \rightarrow (0,\infty)$ given by 
     $$
     \bm G^S\big(x_1, \cdots, x_d \big) \,:= \,x_{m+1} + \cdots + x_d \,.
     $$ 
 The function $G^S := \mathbf{ G}^S \circ \mathfrak R\,$ is  again regular.  Exactly as above, we compute, 
  \begin{align*}
	D_i G^S(x) &=  \1_{\,x_{(m)}>x_i   } +    \frac{\sum_{\ell=m+1}^d  \1_{\,  x_{(\ell)}=x_i } }{\sum_{\ell=1}^d  \, \1_{\, x_{(\ell)}=x_i  }}  \,\,  \1_{\, x_{(m)}=x_i }\,\,, \qquad x \in \bm \Delta^d, \, \,\,\,i = 1, \cdots, d\\
	\Gamma^{G^S}(\cdot) &=   \sum_{\ell=m+1}^{d}   \sum_{k=1}^{m}    \int_0^{ \cdot}   \frac{1}{N_m(\mu(t))}    \mathrm{d}  \Lambda^{(k,\ell)}(t).
\end{align*}
 Thus, $G^S$ is not only regular, but also a {Lyapunov} function. The non-decreasing process $\Gamma^{G^S}(\cdot)$ expresses  the cumulative gains that the additively-generated strategy generates; whenever it sells a stock, this strategy  sells it  at a profit --- the stock has been promoted to the ``major (capitalization) league.''  
 
 It is again simple to see that the additively-generated strategy invests in all assets, provided that $\Gamma^{G^S}(\cdot)$ is not identically equal to zero; while the multiplicatively-generated strategy only invests in the $d-m$ smallest stocks.
\qed
 \end{example}

  \begin{example} [Small stocks, again] 
     \label{small_again}
     Under the setup of Example~\ref{small}, and a bit more generally, 
     consider a function $ \bm{ H} : [0,1]^{d-m} \to \R$, which is regular for the truncated vector process of ranked market weights $(\mu_{(m+1)}(\cdot), \cdots, \mu_{(d)}(\cdot))$; for example,  if it is twice continuously differentiable. 
Then it is clear that   the function   $ \mathbf{ G} :\mathbb W^d \rightarrow (0,\infty)$ given by 
     $$
     \bm G \big(x_1, \cdots, x_d \big) \,:= \, \bm H \big( x_{m+1}, \cdots , x_d \big)
     $$ 
is regular for the full vector process of ranked market weights   $\bm \mu(\cdot) = \transpose{(\mu_{(1)}(\cdot), \cdots, \mu_{(d)}(\cdot))}$ with $\Gamma^{\bm G}(\cdot) =  \Gamma^{\bm H}(\cdot)$, $D_\ell \bm G = 0$ for all $\ell = 1, \cdots, m$, and $D_\ell \bm G = D_\ell \bm H$ for all $\,\ell = m+1, \cdots, d\,.$  Here we write $D \bm H = \transpose{(D \bm H_{m+1}, \cdots, D \bm H_d)}$ in Definition~\ref{def: reg}. 

Hence, by Theorem~\ref{thm: 1b}  the function $G  : =  \bm G \circ \mathfrak R : \Delta^d \to \R\,$, that is,
\begin{align*}
	G( x )\,=\,   \bm H\big(x_{(m+1)}, \cdots, x_{(d)}\big)
\end{align*}
is also regular for the vector process $\mu(\cdot)$ of market weights.
As in \eqref{eq:160117.2} of Example~\ref{ex: smooth2}, we obtain
\begin{align*}
	\Gamma^G(\cdot) &=  \Gamma^{\bm G}(\cdot)
-  \sum_{\ell=1}^{d-1}   \sum_{k=\ell+1}^{d}    \int_0^{ \cdot}   \frac{1}{N_\ell(\mu(t))}    D_\ell \bm G \big( \bm\mu ( t) \big)  \mathrm{d}  \Lambda^{(\ell,k)}(t)   + 
 \sum_{\ell=1}^{d}   \sum_{k=1}^{\ell-1}    \int_0^{ \cdot}   \frac{1}{N_\ell(\mu(t))}   D_\ell \bm G \big( \bm\mu ( t) \big)  \mathrm{d}  \Lambda^{(k,\ell)}(t) \\
	&= \Gamma^{\bm H}(\cdot)
+  \sum_{\ell=1}^{d-1}   \sum_{k=\ell+1}^{d}    \int_0^{ \cdot}   \frac{1}{N_\ell(\mu(t))}  \left(D_k \bm G \big( \bm\mu ( t))  -  D_\ell \bm G \big( \bm\mu ( t) \big) \right)  \mathrm{d}  \Lambda^{(\ell,k)}(t) .
\end{align*}

\bigskip
\noindent
{\it Permutation Invariance:} 
Let us now assume that $\bm H$ is concave, differentiable, and  invariant under permutations of its variables; that is,  
$G$ is a symmetric function of the $d-m$ smallest components of its argument. Then we may assume that for every $x \in [0,1]^{d-m}$ with $x_k = x_\ell$ we have $D_k \bm H(x) =  D_\ell \bm H(x)$, for all $   m+1\le \ell,k  \le d$. Then, in particular, we get $D_k \bm G \big( \bm\mu ( \cdot))  = D_\ell \bm G \big( \bm\mu ( \cdot) \big)$ on the support of $\Lambda^{(\ell,k)}(\cdot)$ for each $k= \ell+1, \cdots, d$ and $\ell=m+1, \cdots, d$. Since also $D_\ell \bm G = 0$ for each $\ell = 1, \cdots, m$, we now have
\begin{align*}
	\Gamma^G(\cdot)
	&=  \Gamma^{\bm H}(\cdot) +  \sum_{\ell=1}^{m}   \sum_{k=m+1}^{d}    \int_0^{ \cdot}   \frac{1}{N_m(\mu(t))} D_k \bm G \big( \bm\mu ( t))   \mathrm{d}  \Lambda^{(\ell,k)}(t).
\end{align*}
Since $\bm H$ is concave  this finite-variation process is non-decreasing, thus   $G$ is a {Lyapunov} function.   If $\bm H$ is nonnegative and $G(\mu(0)) > 0$, then Theorem~\ref{thm: 2} shows now that for some given real number $ T>0$ strong outperformance  exists with respect to the market over the horizon $[0,T]$ if   
 $
\P(\Gamma^{\bm H}(T) > G ( { \mu} (0) )) = 1$.
For example, if $\bm H$ is twice differentiable we have
\begin{align*}
	\Gamma^{\bm H}(\cdot)
	&= -\frac{1}{2}  \sum_{\ell = m+1}^d \sum_{k=m+1}^d \int_0^{ \cdot}    D^2_{\ell k }
\bm H \big( \mu_{(m+1)}(t), \cdots, \mu_{(d)}(t)\big)     \mathrm{d} \big \langle \mu_{(\ell)}, \mu_{(k)} \big \rangle (t) .
\end{align*}
 Section~4 in \citet{Vervuurt:Karatzas:2015} develops in detail  a special case of such a  construction for multiplicatively generating a trading strategy.
\qed
 \end{example}

\section{Concave transformations of semimartingales} 
 \label{sec: 8}

Consider a  function
$G:\bm \Delta^d \to \R$.  The  {\it superdifferential} of $G$ at some point $x \in \bm \Delta^d $,  denoted by $  \partial  G(x)$, is the set of all ``supergradients''   at that point; namely,  the set of all  vectors $\xi \in \R^d$ such that
\begin{align} \label{eq: 7.1} 
	\sum_{i = 1}^d \xi_i (y_i - x_i) \geq G(y) - G(x)  \quad \text{holds for all } y \in \bm \Delta^d .
\end{align}
If $G$ is concave we have $\partial G(x) \neq \emptyset$ for all $x \in \bm \Delta^d_+$.

\subsection{The proof of Theorems~\ref{thm: 1} and \ref{thm: 1b} and Proposition~\ref{pr: 1b}}
\label{sec: App1}

\begin{proof}[Proof of Theorem~\ref{thm: 1}]
We proceed in three steps.

\noindent
{\it Step 1:}  We shall find it it  useful to identify the set ${\bm \Delta}^d_+ \subset \R^{n}_+$ of \eqref{eq: Delta} with  the set 
  \begin{align}
  \label{eq: flat}
         {\bm \Delta}^d_{\flat+} 
     \,:   =   \, \bigg\{\big( x_1, \cdots , x_{d-1}\big)' \in (0,1)^{d-1} :\, \sum_{i =1}^{d-1} x_i <1  \bigg\} \, \subset \mathbb{R}^{d-1}.
   \end{align}
The identification is based on the one-to-one ``projection operator'' $ \mathfrak P$ namely, the mapping $ {\bm \Delta}^d_e \ni (x_1, \cdots, x_d ) \mapsto  (x_1, \cdots, x_{d-1} )\in \R^{d-1}$. In this manner, a real-valued function  $G$ on ${\bm \Delta}^d_+$ or on ${\bm \Delta}^d_e$ is identified with the   function $ G_\flat = G \circ  \mathfrak P^{-1}$ on ${\bm \Delta}^d_{\flat+}$ or on $\,\R^{d-1}$, respectively; and vice-versa. Note that $G$ is concave on ${\bm \Delta}^d_+$ or on ${\bm \Delta}^d_e\,,$ if and only if $G_\flat$ is concave on ${\bm \Delta}^d_{\flat+} $ or on $\R^{d-1}$, respectively.  

\medskip
\noindent
{\it Step 2:} {\it Let us start by imposing either condition \ref{thm: 1i} or \ref{thm: 1ii}}.  We  recall from Theorem~10.4 in \citet{Rockafellar:1970} (see also \citet{WayneState} and \citet{Roberts:Varberg}) that the concave function   $ G_\flat = G \circ  \mathfrak P^{-1}$ is locally Lipschitz on the open set ${\bm \Delta}^d_{\flat +}$ of \eqref{eq: flat} or on $\R^{d-1}$, respectively.   Theorem~VI.8 in \citet{Meyer:1976}, along with the remark on page~222 of \citet{Dellacherie/Meyer:1982}, yields that $G(\mu(\cdot))$ is a semimartingale.

We now let $ DG = \big( D_1 G, \cdots , D_d G\big)' : {\bm \Delta}^d \rightarrow \R^d$ 
denote any measurable ``supergradient'' of $G;$ that is, $DG$ is measurable and satisfies $ D G(x) \in \partial G (x)$ for all $x \in {\bm \Delta}^d_+$ in  Theorem~\ref{thm: 1}\ref{thm: 1i}, and for all $x \in {\bm \Delta}^d_e$ in Theorem~\ref{thm: 1}\ref{thm: 1ii}.
 The {It\^o}-type formula implicit in \eqref{eq: deco}, namely 
\begin{equation}
\label{eq: deco_too}
  G \big( \mu (0) \big) =  G \big( \mu (\cdot) \big) +  \int_0^{ \cdot}  \big\langle D G \big( \mu ( t) \big) ,  \mathrm{d} \mu (t) \big\rangle - \Gamma^G (\cdot)  
\end{equation}
with a continuous, non-decreasing $\Gamma^G (\cdot),$ is established as in \citet{Bouleau:1981:convex, Bouleau:1984}; see also \cite{Grinberg:2013} for an alternative treatment, and   \cite{Aboulaich:Stricker:convex} for  the special case where $G$ is once continuously differentiable.  With the obvious notation $DG_\flat$, we    use  here  the identity $$ \int_0^{ \cdot} \big\langle DG_\flat  (\mu_\flat (t)  ) , \mathrm{d} \mu_\flat (t) \big\rangle =\int_0^{ \cdot} \big\langle DG  (\mu  (t)  ) , \mathrm{d} \mu (t)\big\rangle$$ of stochastic integration  for the process $ \mu_\flat (\cdot) =  \transpose{( \mu_1 (\cdot) , \cdots, \mu_{d-1} (\cdot)  )}.$ 

\medskip
\noindent
{\it Step 3:}  {\it  We place ourselves now under the assumptions of \ref{thm: 1iii}.} For sake of notational simplicity we shall assume here  $ \mathbf{supp\,} (\mu )  = \Delta^d;$   the general case follows in exactly the same manner.   We recall   the stopping time $\mathscr D$ in \eqref{eq: 6}, and note that any component $\mu_i(\cdot)$ with $\mu_i(\mathscr D) = 0$ is absorbed at the origin:   $\mu(\mathscr D+t) = 0$ holds for all $t \geq 0$ on the event $ \{ \mathscr D < \infty\}$ (see  Subsection~\ref{sec: 2a}). We use the notation $\mathfrak{m}: \Omega \to \{ 1, \cdots, d\}$ for the $\sigalgebra{F}(\mathscr D)$--measurable random variable that records the number of assets which have not been absorbed by time $\mathscr D$; namely, the number of all indices $i \in \{1, \cdots, d\}$ such that $\mu_i(\mathscr D)>0$.

Assume    we have shown that 
\begin{equation}
\label{eq: semi_ass}
 G\big( \mu (\cdot   \wedge \mathscr D) \big) \quad \text{ is a semimartingale.}
\end{equation}
Then,  after time $\mathscr D$, the process $ G\big( \mu(\cdot)\big)$ can be identified with a process $\widetilde G\big(\widetilde \mu(\cdot)\big),$  where $\widetilde{\mu}(\cdot)$ takes values in $ {\bm \Delta}^{  \mathfrak m}$, the   domain of a concave function $\widetilde{G}$. An iteration of the argument then yields the statement, since the {It\^o}-type formula in \eqref{eq: deco_too} follows again, exactly as in Step~2. Indeed, as above, $DG$ may denote any measurable supergradient of $G$ on  $\bm \Delta^d_+$. On  $\bm \Delta^d \setminus 
\bm \Delta^d_+$, the concave function $G$ can be identified with a concave function $\widehat G$ on $\bm \Delta^n$ for some $n < d$. Thus, for each $x \in \bm \Delta^d$ and $i = 1, \cdots, d,$  if $x_i \in \{0,1\}$, we can set the $i$-th component of $DG(x)$ to zero (any arbitrary number would work); and if $x_i \in (0,1)$, to the corresponding component of the supergradient of $\widehat G$. 

We still need to justify the claim in \eqref{eq: semi_ass}. Since $G$ is continuous and thus bounded on the the compact set $\bm \Delta^d$, we may assume, without loss of generality, that $G$ is nonnegative. Let $Z(\cdot)$ denote a deflator for  the vector process $\mu (\cdot)$.  Next, we introduce the increasing sequence of stopping times
$$
\mathscr S_n =  \inf \Big\{ t \ge 0 :\, \min_{i = 1, \cdots, d } \, \mu_i (t) < \frac{1}{n}\,\Big\}\,,\qquad n \in \N\,,
$$
satisfying $\lim_{n \uparrow \infty} \mathscr S_n  = \mathscr D$.
Thanks to \eqref{eq: 5} and \ref{thm: 1i}, in particular \eqref{eq: deco_too},  the process $Z( \cdot \wedge \mathscr S_n)) G ( \mu ( \cdot \wedge \mathscr S_n) )$ is a local supermartingale for each $n \in \N$,  and thus $ (  Z(t) G ( \mu ( t  ) ) )_{0 \le t <\mathscr D}$ is a local  supermartingale,  bounded from below. The supermartingale convergence theorem, see Lemma~4.14 in \cite{Ruf_Larsson}, yields that  $ Z ( \cdot  \wedge \mathscr D)  G \big( \mu ( \cdot  \wedge \mathscr D) ) $ is also a  local supermartingale. From this, and from the fact that the reciprocal $1/Z(\cdot \wedge \mathscr D)$ is a semimartingale, the claim in \eqref{eq: semi_ass} follows.
\end{proof}

The proof of Theorem~\ref{thm: 1} shows that every continuous, concave  function $G$ is   regular, and the   $DG$ in the corresponding It\^o formula of \eqref{eq: deco} may be chosen (at least in the set $\bm \Delta^d_+$) to be a measurable supergradient of $G$. 
This observation also motivates the following conjecture.
\begin{rem}[An open question concerning $DG$]
	Assume that a function $G$ is regular and weakly differentiable with gradient $\widetilde {DG}$.  Is it then possible to choose $DG = \widetilde{DG}$ in \eqref{eq: thetaG} and \eqref{eq: deco}?
	\qed
\end{rem}

Concerning a representation of the finite-variation process $\Gamma^G(\cdot)$, the proof of Theorem~\ref{thm: 1} does not yield any deep insights.  The arguments in \citet{Bouleau:1981:convex, Bouleau:1984} yield a representation of $\Gamma^G(\cdot)$ as a limit of mollified second-order terms. 
\begin{rem}[An open question concerning $\Gamma^G(\cdot)$]
In the context of Theorem~\ref{thm: 1} we conjecture that  the process $ - 2  \Gamma^G (\cdot)$ of \eqref{eq: deco} can be written as the sum of the   covariations of the processes ${\bm \vartheta}_i (\cdot)$ as in \eqref{eq: thetaG} and $ \mu_j (\cdot)$ as in \eqref{eq: mw}; namely,  
$$
\Gamma^G (\cdot) = -\frac{1}{\,2\,} \sum_{i=1}^d \sum_{j=1}^d  \big[ {\bm \vartheta}_i , \mu_j \big] (\cdot),
$$
whenever the limits   below exist  in probability, for all $T \geq 0$ and $1 \le i,j \le d\,$:
\[
\big[ {\bm \vartheta}_i , \mu_j \big] (T)= \big( \P \big) \,\,\lim_{N \uparrow \infty} \sum_{n:\, t^{(N)}_n \in \mathbb{D}^{(N)} ,\,\, t^{(N)}_n < T} \Big(  {\bm \vartheta}_i \big( t^{(N)}_{n+1} \big) - {\bm \vartheta}_i \big( t^{(N)}_n \big) \Big)  \Big( \mu_j \big( t^{(N)}_{n+1} \big) - \mu_j \big( t^{(N)}_n \big) \Big). 
\]
 Here $ \big( \mathbb{D}^{(N)}\big)_{N \in \N}$ is a sequence of partitions of $ [0, \infty)$ of the form $0= t^{(N)}_{ 0} < t^{(N)}_{ 1}  < t^{(N)}_2 < \cdots \,,$ with each $    \mathbb{D}^{(N+1)} $ a refinement of  $    \mathbb{D}^{(N)}$, and with mesh $ \big \Vert \mathbb{D}^{(N)} \big \Vert = \max_{n \in \N_0} \{| t^{(N)}_{n+1} - t^{(N)}_{n }|\} $ decreasing all the way  to zero as $N \uparrow \infty$. 
 \qed
 \end{rem}

\begin{proof}[Proof of Theorem~\ref{thm: 1b}]
First, note that \eqref{eq: interiors}  is equivalent to the condition 
\begin{equation*}
\P \big( \mathfrak{R}({ \mu} (t)) \in   {\bm \Delta}^d_+, \,\,\forall \,\, t \geq 0 \big) =1.
\end{equation*} 
Hence, the sufficiency of conditions \ref{thm: 1bi}, \ref{thm: 1biii} here, is a simple corollary of the sufficiency of conditions  \ref{thm: 1i},  \ref{thm: 1ii} in Theorem~\ref{thm: 1} with $G$ replaced by ${\bm G}$, and applied to the $\bm \Delta^d$--valued process $\bm \mu(\cdot) =  \mathfrak{R}({ \mu} (\cdot))$. 

It remains to be argued that  if the function $\bm G$ is   regular  for the vector process $ {\bm \mu} (\cdot)$, then the function $G = {\bm G} \circ \mathfrak R$ is   regular   for the vector process $\mu(\cdot)$.  
Towards this end, we generalize the arguments in Example~\ref{ex: smooth2}. First, in a manner  similar  to \eqref{eq:160128.1}, we recall from Theorem~2.3  in \citet{Banner:Ghomrasni}  the existence of measurable functions ${\bm h}_\ell: \bm \Delta^d  \rightarrow [0,1]$  and of  finite variation processes ${\bm B}_\ell(\cdot)$ with ${\bm B}_\ell(0) = 0$,   such that  we have 
$$
	\bm{\mu}_\ell(\cdot) =  \bm{\mu}_\ell(0) + \int_0^\cdot \sum_{i=1}^d {\bm h}_\ell(\mu (t))  \1_{\{\mu_{(\ell)}(t)  = \mu_i(t)\}} \mathrm{d} \mu_i (t) 
		+ {\bm B}_\ell(\cdot)  
		$$
for all $\,\ell = 1, \cdots, d\,.$ Therefore, we obtain
\begin{align*}
	G(\mu(\cdot)) &=  {\bm G} ( \mathfrak R(\mu(\cdot))) =  {\bm G} ( \bm\mu(\cdot))
		=  {\bm G} ( \bm\mu(0)) +   \int_0^{ \cdot} \sum_{\ell=1}^d D_\ell \bm G \big( \bm\mu ( t) \big) \mathrm{d} \bm \mu_\ell (t) - \Gamma^{\bm G} (\cdot),
\end{align*}
where $D \bm G$ and $\Gamma^{\bm G}(\cdot)$ are as in Definition~\ref{def: reg}; in particular,  $\Gamma^{\bm G}(\cdot)$  is a finite-variation process. 
By analogy with \eqref{eq:160117.1} and \eqref{eq:160117.2}, we define now 
\begin{align*}
	D_i G(x) &:= \sum_{\ell=1}^d \,  {\bm h}_\ell(x) \,  D_\ell \bm G \big( \mathfrak R(x) \big) \, \1_{ \,x_{(\ell)}= x_i }\,, \qquad x \in \mathbf{supp\,} ( \mu)\,, \,\,\,\, i = 1, \cdots, d\,,\\
	\Gamma^G(\cdot) &:=  \Gamma^{\bm G} (\cdot) -  \sum_{\ell=1}^d  \int_0^{ \cdot} D_\ell \bm G \big( \bm\mu ( t) \big)  \mathrm{d}  {\bm B}_\ell (t)\,,
\end{align*} 
and note ${\bm G} ( \bm\mu(0)) = G(\mu(0))$. This yields    \eqref{eq: decor}, thus also the regularity of $G$    for $\mu(\cdot)$.   
\end{proof}

 \begin{proof}[Proof of Proposition~\ref{pr: 1b}] Theorem~\ref{thm: 1} shows that $G$ is a Lyapunov function; its  proof also  reveals that    $DG$ can be chosen  to be a supergradient of $G,$ if \ref{thm: 1i} or \ref{thm: 1ii} hold. If neither \ref{thm: 1i} nor \ref{thm: 1ii} holds,  but \ref{thm: 1iii} does,   we may choose $DG$ to be a supergradient of $G$ in $\bm \Delta^d_+$. In that case, for $x \in \bm \Delta^d \setminus \bm \Delta^d_+$ and $i = 1, \cdots, d$, we   declare $D_i G(x)$ to be the corresponding component of a concave function $\widetilde G$ with domain $\bm \Delta^m$ for some $m<d$ if $x_i \in (0,1),$ and otherwise to be the term $\sum_{j: x_j \in (0,1)} x_j D_j G(x)$. 

Fixing this choice of $DG$ we next note from \eqref{eq: phiG2}
that the non-decrease of $ \Gamma^G (\cdot)$ gives
 $$
{\bm \varphi}_i ( \cdot) \ge G \big( \mu (\cdot) \big) +  D_i G \big( \mu (\cdot) \big)   -  \sum_{j=1}^d  \mu_j (\cdot) \, D_j G  \big( \mu (\cdot) \big) 
$$
so it suffices to show, for every fixed  $\,i = 1, \cdots, d\,$ and $ x \in {\bm \Delta}^d$, the inequality 
\begin{align} \label{eq:150923.2}
G (x)+ D_i G  ( x)    -  \sum_{j=1}^d  x_j    D_j G   ( x)  \ge 0\,.
\end{align}
  If $x_i \in \{0,1\}$ then \eqref{eq:150923.2} follows directly from the nonnegativity of $G$.  Thus, we now consider  the case $ x_i \in (0,1)$, and  let $\mathfrak e^{(i)}\in   \bm \Delta^d$ denote the $i$-th unit  vector of $\mathbb{R}^d$.   Observe that if $x_j = 0$ for some $j = 1, \cdots, d$, then  the $j$-th component of any linear combination of $x$ and  $\mathfrak e^{(i)}$ is also zero.  This fact,  the nonnegativity of $G$, and the property of supergradients given in \eqref{eq: 7.1}, lead   to 
  \begin{align*}
  	0 \leq G\big(ux + (1-u) \mathfrak e^{(i)}\big) & \leq G(x) + \sum_{j=1}^d \, \big((u-1) x_j + (1-u) \mathfrak e^{(i)}_j\big) \, D_j G(x) \\
		&= G(x) + (1-u) D_i G(x) - (1-u)  \sum_{j=1}^d x_j D_j G(x)
  \end{align*}
 for all $u \in (0,1]$. Letting $u \downarrow 0$ yields  \eqref{eq:150923.2},  and thus  the statement.
\end{proof}

\subsection{Two counterexamples}
 
\begin{example}[Lack of deflator in  Theorem~\ref{thm: 1}\ref{thm: 1iii}] \label{ex:Counterexample1}
	A condition, such as the existence of a   deflator in Theorem~\ref{thm: 1}\ref{thm: 1iii}, is  needed for the result to hold. 	
	Even for a one-dimensional semimartingale $X (\cdot)$ taking values in the unit 	interval $ [0,1]$  and   absorbed when it hits one of its 
	endpoints, and  with a  
	concave 
	function $G: [0,1] 
	\to [0,1]$, the process   $G(X (\cdot))$ need  not be a semimartingale.

	For example, let $X$ be a deterministic continuous semimartingale with $X(0) = 1$ and $X(t) = \lim_{s \uparrow 1} X(s) = 0$ for all $t \geq 1$, constructed as follows.  Let $a_n$ be the smallest odd integer  in the interval $[\sqrt{n}, 3 \sqrt{n})$, for all $n \in \N$.  On $[1-1/n, 1-1/(n+1)]$ let $X(\cdot)$ have exactly $a_n$ oscillations between  $1/n$ and $1/(n+1)$, for each $n \in \N$. 
	In particular, $X(1-1/n) = 1/n$ and $X(t) \in [1/(n+1), 1/n]$ for all $t \in [1-1/n,1 - 1/(n+1)]$, for each $n \in \N$.  
Then $X$ is clearly continuous and takes values in the compact interval 
	$[0,1]$.   
	Since the first variation of $X(\cdot)$ is exactly
	\begin{align*}
		\sum_{n \in \N} a_n \left(\frac{1}{n} - \frac{1}{n+1}\right)    \leq  \sum_{n \in \N}  \frac{3 \sqrt{n}}{n^2+n}<\infty, 
	\end{align*}
	the process $X(\cdot)$ is indeed a continuous, deterministic finite-variation semimartingale.

	Now consider the concave and bounded function $\widehat G: [0,1] 
	\rightarrow [0,1]$ with $\widehat G( x  )  := \sqrt{x}
	$. Then the  first variation of $  \widehat G(X(\cdot))$ is exactly
	 	$$
				\sum_{n \in \N} a_n \left(\sqrt{\frac{1}{n}} - \sqrt{\frac{1}{n+1}}\,\right)       		 \geq   
				\sum_{n \in \N} \left(1 - \sqrt{\frac{n}{n+1}}\,\right) \ge \sum_{n \in \N} \frac{\bm \kappa}{n} = \infty
		$$
			for some $\bm \kappa > 0$, where the last inequality follows from {l'H\^opital}'s rule.  Thus $\widehat G(X(\cdot))$ is deterministic, but of infinite variation and thus not a semimartingale.  	It follows that, without further assumptions, a concave and continuous transformation  defined on a convex set $[0,1]$, of a continuous semimartingale taking values in $[0,1]$, is not necessarily a semimartingale.
	
	To put this example in the context of Theorem~\ref{thm: 1}, just set $d=2$, $\mu_1(\cdot) := X(\cdot) $, 
	and $\mu_2(\cdot) := 1- \mu_1(\cdot)$. Then, there exists no deflator for $\mu(\cdot)$ and the concave and continuous function $G(x_1, x_2) := \sqrt{x_1}$, for all $(x_1, x_2) \in \bm \Delta^2$ is indeed not a regular function for the process $(\mu_1(\cdot), \mu_2(\cdot))$. 
	\qed
\end{example}

\begin{example}[Existence of deflator in  Theorem~\ref{thm: 2}, but lack of regularity] \label{ex:Counterexample2}
We now modify Example~\ref{ex:Counterexample1} to obtain a setup in which a deflator for the vector process $\mu(\cdot)$ exists, the function $\bm G: \mathbb W^d \rightarrow [0,1]$ is continuous and concave, but $\bm G$ is not regular for $\bm \mu(\cdot) = \mathfrak R(\mu(\cdot))$ in the notation of \eqref{eq: Wn} and \eqref{eq: R}.

To this end, set $d = 2$ and let $B(\cdot)$ denote a Brownian motion starting at $B(0)=1$, and stopped when hitting 0 or 2. We set $\mu_1(\cdot) := B(\cdot)/2$ and $\mu_2(\cdot) := 1-B(\cdot)/2 = 1 - \mu_1(\cdot)$. Since $\mu_1(\cdot)$ and $\mu_2(\cdot)$ are martingales, there exists a deflator for the vector process $\mu(\cdot)$; indeed, $Z(\cdot) \equiv 1$ will serve as one. Next, consider the function $\, \bm G(x_1, x_2) := \sqrt{  x_1 - x_2 \,}\,$ 
for all $(x_1,x_2) \in \mathbb W^2 = \mathbf{supp\,} \bm \mu$. Clearly, $\bm G$ is concave and continuous on $\mathbb W^2$. However, by Lemma~\ref{L: sqrtBM} below the process $\bm G(\bm \mu(\cdot)) = \sqrt{|1-B(\cdot)|}$ is not a semimartingale, thus $\bm G$ is  not regular for $\bm \mu(\cdot)$. \qed
\end{example}

\begin{lem}[Square root of Brownian motion] \label{L: sqrtBM}
	Let $W(\cdot)$ denote a Brownian motion starting in zero and $\tau$ a strictly positive stopping time. Then  the process $\sqrt{|W(\cdot \wedge \tau)|}$ is not a semimartingale.
\end{lem}
\begin{proof}
	Of course, the function $f: \R \ni x \mapsto \sqrt{|x|}$ is not the difference of two convex functions, which would let us conclude  from the results in \citet{Cinlar1980}, at least formally. For sake of completeness we  provide here a direct proof.  Note that the quadratic variation of $\,2 f(W(\cdot \wedge \tau))\,$ can be bounded from below by the quadratic variation of the semimartingales
	\[
	2 \sqrt{\varepsilon \vee |W(\cdot \wedge \tau)|}, \qquad \varepsilon > 0.
	\]
	Thanks to It\^o's formula, their quadratic variation  is $\int_0^{\cdot \wedge \tau} \1_{\{|W(t)| > \varepsilon\}} 1/|W(t)| \mathrm d t$, for each $\varepsilon>0$.  Thus, the quadratic variation of $2 f(W(\cdot \wedge \tau))$ is at least $\int_0^{\cdot \wedge \tau} \1_{\{W(t) \neq 0\}} 1/|W(t)| \mathrm d t$. An application of the occupation time formula, in conjunction with the continuity of the local time of $W$ then shows that  $f(W(\cdot \wedge \tau))$ has infinite quadratic variation and thus, cannot be a semimartingale.
\end{proof}

For   results in a similar vein, see \citet{Cinlar1980}, especially Theorems~5.8 and 5.9, and \citet{MU_semimartingale}.

\section{Conclusion}
\label{sec: 9}
 \citet{F_generating, Fernholz:2001, Fe} provides a systematic approach to generate trading strategies that can be implemented without the need of heavy statistical estimates and whose performance in a frictionless market can be guaranteed by suitable and weak assumption on the market's volatility structure. 
 The present paper takes a systematic approach to functional generation and makes the following three contributions.
 \begin{enumerate}
 	\item Introduces an alternative, ``additive'' way    to generate trading strategies functionally, and compares it to E.R.~Fernholz'  ``multiplicative''   functional  generation of trading strategies. Given a sufficiently large time horizon $T_*>0$ and suitable conditions on the volatility structure of the market, the multiplicative version yields, for each $T>T_*\,,$ a portfolio that strongly outperforms the market on $[0,T]$; this portfolio, however, depends on $T$. By contrast, the additive version yields a {\it single} trading strategy that outperforms the market   over all horizons $[0,T]$  for $T \geq T_*$. 
	\item Extends the class of functions that generate a trading strategy. This paper introduces the notion of a regular function. Such a function can generate a trading strategy. Modulo necessary technical conditions on the boundary behavior, concave functions are shown to be regular. This weakens the  assumption of twice continuous differentiability, normally used in the literature of the subject, and  provides a unified framework for standard and rank-based generation, a long-standing open issue.
	\item Weakens the assumptions on the market model. Functional generation is shown to work in markets where   asset prices are continuous semimartingales that also might completely devaluate. Moreover, major technical assumptions in rank-based generation are removed; for example, it is not necessary anymore to exclude models where the times when two asset prices are identical have strictly positive Lebesgue measure. 
 \end{enumerate}

\bibliography{aa_bib}{}
\bibliographystyle{apalike}

 \end{document}